\documentclass[reprint,floatfix,nofootinbib,superscriptaddress,aps,pra]{revtex4-1}

\usepackage[utf8]{inputenc}
\usepackage[T1]{fontenc}
\usepackage{microtype} 
\usepackage{textcomp}

\usepackage{graphicx}
\usepackage[caption=false]{subfig}
\usepackage{bm}
\usepackage{amsmath}
\usepackage{amssymb}
\usepackage{stmaryrd}
\usepackage{braket}
\usepackage{enumerate}
\usepackage{amsthm}
\newtheorem*{theorem}{Theorem}

\usepackage{tikz}
\makeatletter 
\RequirePackage{xargs,ifthen,xstring,xparse,etoolbox,mathtools,pgfmath}
\RequirePackage{environ} 
\usetikzlibrary{cd,decorations.pathreplacing,calc,positioning,fit,shapes.symbols,decorations.pathmorphing,shapes.misc,backgrounds,decorations.markings,math}

\newlength{\myl}
\newlength{\myh}
\newlength{\myd}

\newcounter{aaa}
\tikzset{
  apply/.style args={#1 except on segments #2}{postaction={
      /utils/exec={
        \@for\mattempa:=#2\do{\csdef{aaa@\mattempa}{}}
        \setcounter{aaa}{0}
      },
      decorate,decoration={show path construction,
        moveto code={},
        lineto code={
          \stepcounter{aaa}
          \ifcsdef{aaa@\theaaa}{}{
            \path[#1] (\tikzinputsegmentfirst) -- (\tikzinputsegmentlast);
          }
        },
        curveto code={
          \stepcounter{aaa}
          \ifcsdef{aaa@\theaaa}{}{
            \path [#1] (\tikzinputsegmentfirst) .. controls
            (\tikzinputsegmentsupporta) and (\tikzinputsegmentsupportb)
            ..(\tikzinputsegmentlast);
          }
        },
        closepath code={
          \stepcounter{aaa}
          \ifcsdef{aaa@\theaaa}{}{
            \path [#1] (\tikzinputsegmentfirst) -- (\tikzinputsegmentlast);
          }
        },
      },
    },
  },
}
\newcommand*{\IfInList}[2]{%
	\gdef\memory{0}
	\edef\arg{#1}
	 \foreach \q in #2 {%
	 	\ifthenelse{\q=\arg}{%
	 		\gdef\memory{1}
	 	}{}
	 }
	\ifthenelse{\memory=1} {%
		\expandafter\@firstoftwo
	}{%
		\expandafter\@secondoftwo
	}
}
\NewEnviron{quantikz}[1][]{%
  \def\@temp{\tikzcd@[#1]\BODY}%
  \expandafter\@temp\endtikzcd
}
\def\temp{&} \catcode`&=\active \let&=\temp
\protected\def\vvv#1{\leavevmode\bgroup\vbox\bgroup\xvvv#1\relax}
\def\xvvv{\afterassignment\xxvvv\let\tmp= }
\def\xxvvv{%
\ifx\tmp\relax\egroup\egroup\let\next\relax
 \else
\hbox to 1.1em{\hfill\tmp\hfill}
\let\next\xvvv\fi
\next}
\long\def\ifnodedefined#1#2#3{%
    \@ifundefined{pgf@sh@ns@#1}{#3}{#2}%
}

\pgfmathsetmacro\MathAxis{height("$\vcenter{}$")}

\DeclareExpandableDocumentCommand{\gate}{O{}O{1.5pt}O{1.5pt}m}{
	|[inner sep=4pt,minimum width=#2,minimum height=#3]|%
 	\edef\n{\the\pgfmatrixcurrentrow} 
 	\edef\m{\the\pgfmatrixcurrentcolumn} 
 	\edef\options{row=\n,col=\m,#1}
 	\def\toswap{0}%
 	\def\DisableMinSize{0}%
 	\pgfkeys{/quantikz,wires=1,style=,label style=,braces=}%
  	\pgfkeys{/quantikz,#1}%
 	\pgfkeysgetvalue{/quantikz/wires}{\quantwires}
 	\pgfkeysgetvalue{/quantikz/style}{\a}
 	\pgfkeysgetvalue{/quantikz/label style}{\b}
 	\pgfkeysgetvalue{/quantikz/cwires}{\mylist}
 	\pgfkeysgetvalue{/quantikz/nwires}{\nowires}
 	\pgfkeysgetvalue{/quantikz/bundle}{\bundle}
 	\ifthenelse{\toswap=1}{
 		\def\quantwires{2}
 		\phantom{wide}
 		\settowidth{\myl}{$wide$}
 		\settoheight{\myh}{$wide$}
 		\settodepth{\myd}{$wide$}
 	}{
 		\settowidth{\myl}{$#4$}
 		\ifthenelse{\DisableMinSize=1}{
 			\phantom{U}
 			\settoheight{\myh}{$U$}
 			\settodepth{\myd}{$U$}
 		}{
 			\phantom{#4}
 			\settoheight{\myh}{$#4$}
 			\settodepth{\myd}{$#4$}
 		}
 	}
 	\IfInList{1}{\mylist}{\cw}{\IfInList{1}{\nowires}{}{\IfInList{1}{\bundle}{\qwbundle[alternate]{}}{\qw}}}
 	\edef\k{\the\numexpr\n+\quantwires-1\relax}
 	\edef\mn{\the\numexpr\m-1\relax}
 	\ifthenelse{\quantwires=1}{}{
 	\foreach \i in {\the\numexpr\n+1\relax,...,\k} {
 	    \edef\newcom{\noexpand\vcwhexplicit{\i-\m}{\i-\mn}}
       	\edef\newcomb{\noexpand\vqwexplicit{\i-\m}{\i-\mn}}
       	\edef\newcomc{\noexpand\vqbundleexplicit{\i-\m}{\i-\mn}}
 	 	\edef\val{\the\numexpr\i+1-\n\relax}
 		\IfInList{\val}{\mylist}{\newcom}{\IfInList{\val}{\nowires}{}{\IfInList{\val}{\bundle}{\newcomc}{\newcomb}}}
 		\globaldefs=1
 		\edef\dotikzset{\noexpand\tikzset{row \i\space column \m/.append style={minimum width={max(\the\myl+8pt,#2)}}}}%
 		\dotikzset%
 		\edef\undotikzset{\noexpand\tikzset{row \i\space column \m/.style={}}}%
      	\expandafter\pgfutil@g@addto@macro\expandafter\tikzcd@atendglobals\expandafter{\undotikzset}%
 	}
 		   \globaldefs=1%
 		   \edef\dotikzset{\noexpand\tikzset{row \k\space column \m/.append style={minimum height={max(\the\myh+\the\myd+8pt,#3)}}}}%
 		   \dotikzset%
 		   \globaldefs=0%
 	}
	\expandafter\expandafter\expandafter\expandafter\expandafter\expandafter\expandafter\pgfutil@g@addto@macro\expandafter\expandafter\expandafter\expandafter\expandafter\expandafter\expandafter\tikzcd@atendsavedpaths\expandafter\expandafter\expandafter\expandafter\expandafter\expandafter\expandafter{%
		\expandafter\expandafter\expandafter\expandafter\expandafter\expandafter\expandafter\gate@end\expandafter\expandafter\expandafter\expandafter\expandafter\expandafter\expandafter{\expandafter\expandafter\expandafter\a\expandafter\expandafter\expandafter}\expandafter\expandafter\expandafter{\expandafter\b\expandafter}\expandafter{\options}{#4}
	}
}

\newcommand{\gate@end}[4]{
	\pgfkeys{/quantikz,wires=1}
	\def\toswap{0}
	\def\DisableMinSize{0}
 	\pgfkeys{/quantikz,#3}
	\pgfkeysgetvalue{/quantikz/wires}{\quantwires}
	\pgfkeysgetvalue{/quantikz/row}{\row}
	\pgfkeysgetvalue{/quantikz/col}{\col}
	\ifthenelse{\toswap=1}{\def\quantwires{2}}{}
	\xdef\LoopGG{}
	\foreach \n in  {\row,...,\the\numexpr\row+\quantwires-1\relax} {
	\ifnodedefined{\tikzcdmatrixname-\n-\col}{
		\xdef\LoopGG{\LoopGG(\tikzcdmatrixname-\n-\col)}
		}{}
	}
	\ifthenelse{\toswap=1}{
		\node (group\tikzcdmatrixname-\row-\col) [fit=\LoopGG,operator,inner sep=0pt,#1] {\hphantom{Wide}};
		\draw [thickness] (group\tikzcdmatrixname-\row-\col.west|-\tikzcdmatrixname-\row-\col.center) to[out=0,in=180] (group\tikzcdmatrixname-\row-\col.east|-\tikzcdmatrixname-\the\numexpr\row+1\relax-\col.center);
		\draw [line width=3pt,white,shorten >=0.9pt,shorten <=0.9pt] (group\tikzcdmatrixname-\row-\col.east|-\tikzcdmatrixname-\row-\col.center) to[out=180,in=0] (group\tikzcdmatrixname-\row-\col.west|-\tikzcdmatrixname-\the\numexpr\row+1\relax-\col.center);
		\draw [thickness] (group\tikzcdmatrixname-\row-\col.east|-\tikzcdmatrixname-\row-\col.center) to[out=180,in=0] (group\tikzcdmatrixname-\row-\col.west|-\tikzcdmatrixname-\the\numexpr\row+1\relax-\col.center);
	}{
	  \node (group\tikzcdmatrixname-\row-\col) [fit=\LoopGG,operator,inner sep=0pt,label={[gg label,#2]$#4$},#1] {\hphantom{$#4$}};
	}
}

\DeclareExpandableDocumentCommand{\ghost}{O{}O{0pt}O{1.5pt}m}{
	|[inner ysep=4pt,minimum width=#2,minimum height=#3]| \vphantom{#4}
}

\newcommand\slice[2][]{%
	\pgfkeys{/quantikz,wires=1,style=,label style=,braces=}
	\pgfkeys{/quantikz,#1}%
 	\edef\options{\pgfkeysvalueof{/quantikz/style}}
 	\edef\opts{\pgfkeysvalueof{/quantikz/label style}}
	\edef\n{\the\pgfmatrixcurrentcolumn}
	\expandafter\expandafter\expandafter\expandafter\expandafter\expandafter\expandafter\pgfutil@g@addto@macro\expandafter\expandafter\expandafter\expandafter\expandafter\expandafter\expandafter\tikzcd@atendslices\expandafter\expandafter\expandafter\expandafter\expandafter\expandafter\expandafter{%
		\expandafter\expandafter\expandafter\expandafter\expandafter\expandafter\expandafter\slice@end\expandafter\expandafter\expandafter\expandafter\expandafter\expandafter\expandafter{\expandafter\expandafter\expandafter\n\expandafter\expandafter\expandafter}\expandafter\expandafter\expandafter{\expandafter\options\expandafter}\expandafter{\opts}{#2}%
	}}
\newcommand{\slice@end}[4]{
	\edef\top{($1/2*(\tikzcdmatrixname-col#1.east |- \tikzcdmatrixname-row1.north)+1/2*(\tikzcdmatrixname-col\the\numexpr#1+1\relax.west |- \tikzcdmatrixname-row1.north)$)}
	\edef\bottom{($1/2*(\tikzcdmatrixname-col#1.east |- \tikzcdmatrixname-row\the\pgfmatrixcurrentrow.south)+1/2*(\tikzcdmatrixname-col\the\numexpr#1+1\relax.west |- \tikzcdmatrixname-row\the\pgfmatrixcurrentrow.south)+(0,-3pt)$)}
	\expandafter\expandafter\expandafter\make@slice\expandafter\expandafter\expandafter{\expandafter\top\expandafter}\expandafter{\bottom}{#4}{#2}{#3}
}
\newcommand{\make@slice}[5]{
	\draw[slice,#4] #1 to node[pos=0,inner sep=4pt,anchor=south,color=black,#5] {#3} #2;
}
\newcommand{\sliceallr}{
	\edef\sstyle{\pgfkeysvalueof{/tikz/slice style}}
	\edef\slstyle{\pgfkeysvalueof{/tikz/slice label style}}
	\foreach \n in  {2,...,\the\numexpr\pgfmatrixcurrentcolumn-1-\pgfkeysvalueof{/tikz/remove end slices}\relax} {
		\edef\col{\the\numexpr\n-1\relax}
		\edef\title{\pgfkeysvalueof{/tikz/slice titles}}
		\expandafter\expandafter\expandafter\expandafter\expandafter\expandafter\expandafter\slice@end\expandafter\expandafter\expandafter\expandafter\expandafter\expandafter\expandafter{\expandafter\expandafter\expandafter\n\expandafter\expandafter\expandafter}\expandafter\expandafter\expandafter{\expandafter\sstyle\expandafter}\expandafter{\slstyle}{\title}
	}
}
\newcommand{\sliceallvr}{
	\edef\sstyle{\pgfkeysvalueof{/tikz/slice style}}
	\edef\slstyle{\pgfkeysvalueof{/tikz/slice label style}}
	\foreach \n in  {2,...,\the\numexpr\pgfmatrixcurrentcolumn-1-\pgfkeysvalueof{/tikz/remove end slices}\relax} {
		\edef\col{\the\numexpr\n-1\relax}
		\edef\title{\vvv{\pgfkeysvalueof{/tikz/slice titles}}}
		\expandafter\expandafter\expandafter\expandafter\expandafter\expandafter\expandafter\slice@end\expandafter\expandafter\expandafter\expandafter\expandafter\expandafter\expandafter{\expandafter\expandafter\expandafter\n\expandafter\expandafter\expandafter}\expandafter\expandafter\expandafter{\expandafter\sstyle\expandafter}\expandafter{\slstyle}{\title}
	}
}

\newcommand\lstick[2][]{%
	\pgfkeys{/quantikz,wires=1,style=,label style=,braces=}
	\pgfkeys{/quantikz,#1}%
	\edef\newoptions{row=\the\pgfmatrixcurrentrow,col=\the\pgfmatrixcurrentcolumn,#1}
	\pgfkeysgetvalue{/quantikz/label style}{\options}
	\pgfkeysgetvalue{/quantikz/braces}{\opts}
	\expandafter\expandafter\expandafter\expandafter\expandafter\expandafter\expandafter\pgfutil@g@addto@macro\expandafter\expandafter\expandafter\expandafter\expandafter\expandafter\expandafter\tikzcd@atendsavedpaths\expandafter\expandafter\expandafter\expandafter\expandafter\expandafter\expandafter{%
		\expandafter\expandafter\expandafter\expandafter\expandafter\expandafter\expandafter\groupinput@end\expandafter\expandafter\expandafter\expandafter\expandafter\expandafter\expandafter{\expandafter\expandafter\expandafter\newoptions\expandafter\expandafter\expandafter}\expandafter\expandafter\expandafter{\expandafter\options\expandafter}\expandafter{\opts}{#2}%
	}
}
\newcommand\rstick[2][]{%
	\pgfkeys{/quantikz,wires=1,style=,label style=,braces=}
	\pgfkeys{/quantikz,#1}%
	\edef\newoptions{row=\the\pgfmatrixcurrentrow,col=\the\pgfmatrixcurrentcolumn,#1}
	\pgfkeysgetvalue{/quantikz/label style}{\options}
	\pgfkeysgetvalue{/quantikz/braces}{\opts}
	\expandafter\expandafter\expandafter\expandafter\expandafter\expandafter\expandafter\pgfutil@g@addto@macro\expandafter\expandafter\expandafter\expandafter\expandafter\expandafter\expandafter\tikzcd@atendsavedpaths\expandafter\expandafter\expandafter\expandafter\expandafter\expandafter\expandafter{%
		\expandafter\expandafter\expandafter\expandafter\expandafter\expandafter\expandafter\groupoutput@end\expandafter\expandafter\expandafter\expandafter\expandafter\expandafter\expandafter{\expandafter\expandafter\expandafter\newoptions\expandafter\expandafter\expandafter}\expandafter\expandafter\expandafter{\expandafter\options\expandafter}\expandafter{\opts}{#2}%
	}
}

\newcommand\midstick[2][]{%
	\hphantom{\text{gg#2gg}}\ 
	\pgfkeys{/quantikz,wires=1,style=,label style=,braces=}
	\pgfkeys{/quantikz,#1}%
	\edef\newoptions{row=\the\pgfmatrixcurrentrow,col=\the\pgfmatrixcurrentcolumn,#1}
	\pgfkeysgetvalue{/quantikz/wires}{\quantwires}
	\pgfkeysgetvalue{/quantikz/label style}{\options}
	\pgfkeysgetvalue{/quantikz/braces}{\opts}
	 \edef\n{\the\pgfmatrixcurrentrow} 
 	\edef\m{\the\pgfmatrixcurrentcolumn} 
	\edef\mn{\the\numexpr\m-1\relax}
	\edef\k{\the\numexpr\n+\quantwires-1\relax}
	\settowidth{\myl}{gg#2gg}
	\ifthenelse{\quantwires=1}{}{
	 	\foreach \i in {\the\numexpr\n+1\relax,...,\k} {
 	 	\edef\val{\the\numexpr\i+1-\n\relax}
 		\globaldefs=1
 		\edef\dotikzset{\noexpand\tikzset{row \i\space column \m/.append style={minimum width={\the\myl}}}}
 		\dotikzset
 		\edef\undotikzset{\noexpand\tikzset{row \i\space column \m/.style={}}}
      	\expandafter\pgfutil@g@addto@macro\expandafter\tikzcd@atendglobals\expandafter{\undotikzset}
      	\globaldefs=0
 	}
 	}
	\expandafter\expandafter\expandafter\expandafter\expandafter\expandafter\expandafter\pgfutil@g@addto@macro\expandafter\expandafter\expandafter\expandafter\expandafter\expandafter\expandafter\tikzcd@atendsavedpaths\expandafter\expandafter\expandafter\expandafter\expandafter\expandafter\expandafter{%
		\expandafter\expandafter\expandafter\expandafter\expandafter\expandafter\expandafter\groupoutput@mid\expandafter\expandafter\expandafter\expandafter\expandafter\expandafter\expandafter{\expandafter\expandafter\expandafter\newoptions\expandafter\expandafter\expandafter}\expandafter\expandafter\expandafter{\expandafter\options\expandafter}\expandafter{\opts}{#2}%
	}
}
\newcommand{\groupinput@end}[4]{
	\pgfkeys{/quantikz,wires=1}
 	\pgfkeys{/quantikz,#1}%
	\pgfkeysgetvalue{/quantikz/wires}{\quantwires}
	\pgfkeysgetvalue{/quantikz/row}{\row}
	\pgfkeysgetvalue{/quantikz/col}{\col}
	\xdef\LoopGI{}
	\foreach \n in  {\row,...,\the\numexpr\row+\quantwires-1\relax} {
	\ifnodedefined{\tikzcdmatrixname-\n-\col}{
		\xdef\LoopGI{\LoopGI(\tikzcdmatrixname-\n-\col)} 
		}{}
		}
	\ifthenelse{\quantwires=1} {
		\node (ingr-\row) [fit=\LoopGI, inner sep=0pt,label={[align=center,#2]left:#4}] {};
	}{
	\node (ingr-\row) [fit=\LoopGI, inner sep=0pt] {};
	\draw[dm,#3] ($(ingr-\row.north west)+(-0.1cm,0.1cm)$) to node[midway,align=center,anchor=east,xshift=-0.1cm,#2] {#4} ($(ingr-\row.south west)+(-0.1cm,-0.1cm)$);
	}
} %
\newcommand{\groupoutput@end}[4]{
	\pgfkeys{/quantikz,wires=1}
 	\pgfkeys{/quantikz,#1}%
	\pgfkeysgetvalue{/quantikz/wires}{\quantwires}
	\pgfkeysgetvalue{/quantikz/row}{\row}
	\pgfkeysgetvalue{/quantikz/col}{\col}
	\xdef\LoopGO{}
	\foreach \n in  {\row,...,\the\numexpr\row+\quantwires-1\relax} {
		\ifnodedefined{\tikzcdmatrixname-\n-\col}{
			\xdef\LoopGO{\LoopGO(\tikzcdmatrixname-\n-\col)} 
		}}
		\ifthenelse{\quantwires=1} {
		\node (outgr-\row) [fit=\LoopGO, inner sep=0pt,label={[align=center,#2]right:#4}] {};
	}{
	\node (outgr-\row) [fit=\LoopGO, inner sep=0pt] {};
	\draw[dd,#3] ($(outgr-\row.north east)+(0.1cm,0.1cm)$) to node[midway,align=center,anchor=west,xshift=0.1cm,#2] {#4} ($(outgr-\row.south east)+(0.1cm,-0.1cm)$);
	}
} %
\newcommand{\groupoutput@mid}[4]{
	\pgfkeys{/quantikz,wires=1}
	\def\leftbrace{1}\def\rightbrace{1}
 	\pgfkeys{/quantikz,#1}%
	\pgfkeysgetvalue{/quantikz/wires}{\quantwires}
	\pgfkeysgetvalue{/quantikz/row}{\row}
	\pgfkeysgetvalue{/quantikz/col}{\col}
	\xdef\LoopGO{}
	\foreach \n in  {\row,...,\the\numexpr\row+\quantwires-1\relax} {
		\ifnodedefined{\tikzcdmatrixname-\n-\col}{
			\xdef\LoopGO{\LoopGO(\tikzcdmatrixname-\n-\col)} 
		}}
		\ifthenelse{\quantwires=1} {
		\node (midgr-\row-\col) [fit=\LoopGO, inner sep=0pt,label={[align=center,#2]#4}] {};
	}{
	\node (midgr-\row-\col) [fit=\LoopGO, inner sep=0pt,label={[anchor=mid,#2]center:#4}] {};
	\ifthenelse{\rightbrace=1}{
	\draw[dm,#3] ($(midgr-\row-\col.north east)+(-0.1cm,0.05cm)$) to ($(midgr-\row-\col.south east)+(-0.1cm,-0.05cm)$);
	}{}
	\ifthenelse{\leftbrace=1}{
	\draw[dd,#3] ($(midgr-\row-\col.north west)+(0.1cm,0.05cm)$) to ($(midgr-\row-\col.south west)+(0.1cm,-0.05cm)$);
	}{}
	}
} %
\newcommand\gateinput[2][]{%
	\pgfkeys{/quantikz,wires=1,style=,label style=,braces=}%
	\pgfkeys{/quantikz,#1}%
	\edef\newoptions{row=\the\pgfmatrixcurrentrow,col=\the\pgfmatrixcurrentcolumn,#1}
 	\pgfkeysgetvalue{/quantikz/label style}{\options}
 	\pgfkeysgetvalue{/quantikz/braces}{\opts}%
	\expandafter\expandafter\expandafter\expandafter\expandafter\expandafter\expandafter\pgfutil@g@addto@macro\expandafter\expandafter\expandafter\expandafter\expandafter\expandafter\expandafter\tikzcd@atendlabels\expandafter\expandafter\expandafter\expandafter\expandafter\expandafter\expandafter{%
		\expandafter\expandafter\expandafter\expandafter\expandafter\expandafter\expandafter\mginput@end\expandafter\expandafter\expandafter\expandafter\expandafter\expandafter\expandafter{\expandafter\expandafter\expandafter\newoptions\expandafter\expandafter\expandafter}\expandafter\expandafter\expandafter{\expandafter\options\expandafter}\expandafter{\opts}{#2}%
	}
}
\newcommand\gateoutput[2][]{%
	\pgfkeys{/quantikz,wires=1,style=,label style=,braces=}
	\pgfkeys{/quantikz,#1}%
	\edef\newoptions{row=\the\pgfmatrixcurrentrow,col=\the\pgfmatrixcurrentcolumn,#1}
 	\pgfkeysgetvalue{/quantikz/label style}{\options}
 	\pgfkeysgetvalue{/quantikz/braces}{\opts}
	\expandafter\expandafter\expandafter\expandafter\expandafter\expandafter\expandafter\pgfutil@g@addto@macro\expandafter\expandafter\expandafter\expandafter\expandafter\expandafter\expandafter\tikzcd@atendlabels\expandafter\expandafter\expandafter\expandafter\expandafter\expandafter\expandafter{%
		\expandafter\expandafter\expandafter\expandafter\expandafter\expandafter\expandafter\mgoutput@end\expandafter\expandafter\expandafter\expandafter\expandafter\expandafter\expandafter{\expandafter\expandafter\expandafter\newoptions\expandafter\expandafter\expandafter}\expandafter\expandafter\expandafter{\expandafter\options\expandafter}\expandafter{\opts}{#2}%
	}
}
\newcommand{\mginput@end}[4]{
	\pgfkeys{/quantikz,wires=1}
 	\pgfkeys{/quantikz,#1}%
 	\edef\quantwires{\pgfkeysvalueof{/quantikz/wires}}
	\pgfkeysgetvalue{/quantikz/row}{\row}
	\pgfkeysgetvalue{/quantikz/col}{\col}
\xdef\cell{group\tikzcdmatrixname-1-\col}
\foreach \n in {\row,...,1} {%
	\ifnodedefined{group\tikzcdmatrixname-\n-\col}{%
    	\xdef\cell{group\tikzcdmatrixname-\n-\col}
    	\breakforeach
	}{}
}
\ifthenelse{\quantwires=1}{%
	\node at ($(\cell.west |- \tikzcdmatrixname-\row-\col.west)+(0,0cm)$)[leftinternal,#2]{#4};
}{%
\draw[dd,#3] ($(\cell.west |- \tikzcdmatrixname-\row-\col.west)+(0.1cm,0.1cm)$) to node[leftinternal,midway,#2] {#4} ($(\cell.west |- \tikzcdmatrixname-\the\numexpr\row+\quantwires-1\relax-\col.west)+(0.1cm,-0.1cm)$);
}
} %
\newcommand{\mgoutput@end}[4]{%
	\pgfkeys{/quantikz,wires=1}
 	\pgfkeys{/quantikz,#1}%
 	\edef\quantwires{\pgfkeysvalueof{/quantikz/wires}}
	\pgfkeysgetvalue{/quantikz/row}{\row}
	\pgfkeysgetvalue{/quantikz/col}{\col}
\xdef\cell{group\tikzcdmatrixname-1-\col}
\foreach \n in {\row,...,1} {%
	\ifnodedefined{group\tikzcdmatrixname-\n-\col}{%
    	\xdef\cell{group\tikzcdmatrixname-\n-\col}
    	\breakforeach
	}{}
}
\ifthenelse{\quantwires=1}{%
	\node at ($(\cell.east |- \tikzcdmatrixname-\row-\col.east)+(0,0cm)$)[rightinternal,#2]{#4};
}{%
\draw[dm,#3] ($(\cell.east |- \tikzcdmatrixname-\row-\col.east)+(-0.1cm,0.1cm)$) to node[rightinternal,midway,#2] {#4} ($(\cell.east |- \tikzcdmatrixname-\the\numexpr\row+\quantwires-1\relax-\col.east)+(-0.1cm,-0.1cm)$);
}
} %

\newcommand\wave[1][]{%
	\edef\n{\the\pgfmatrixcurrentrow}
	\expandafter\pgfutil@g@addto@macro\expandafter\tikzcd@atendslices\expandafter{%
		\expandafter\wave@end\expandafter{\n}{#1}%
	}
}
\newcommand{\wave@end}[2]{
	\node (wave-#1) [fit=(\tikzcdmatrixname-row#1),wave,#2] {};
}

\DeclareDocumentCommand{\makeebit}{O{-45}O{}m}{
	\arrow[arrows,line cap=round,to path={(\tikztostart) -- ($(\tikztostart)!{0.5/cos(#1)}!#1:(\tikztotarget)$) node [anchor=east,style={#2}]{#3} -- (\tikztotarget) }]{d}
}

\newcommand\gategroup[2][]{%
	\pgfkeys{/quantikz,wires=1,style=,label style=,braces=}
	\pgfkeys{/quantikz,#1}%
	\edef\newoptions{row=\the\pgfmatrixcurrentrow,col=\the\pgfmatrixcurrentcolumn,#1}
 	\pgfkeysgetvalue{/quantikz/style}{\options}
 	\pgfkeysgetvalue{/quantikz/label style}{\opts}
	\expandafter\expandafter\expandafter\expandafter\expandafter\expandafter\expandafter\pgfutil@g@addto@macro\expandafter\expandafter\expandafter\expandafter\expandafter\expandafter\expandafter\tikzcd@atendlabels\expandafter\expandafter\expandafter\expandafter\expandafter\expandafter\expandafter{%
		\expandafter\expandafter\expandafter\expandafter\expandafter\expandafter\expandafter\gategroup@end\expandafter\expandafter\expandafter\expandafter\expandafter\expandafter\expandafter{\expandafter\expandafter\expandafter\newoptions\expandafter\expandafter\expandafter}\expandafter\expandafter\expandafter{\expandafter\options\expandafter}\expandafter{\opts}{#2}%
	}
}
\newcommand{\gategroup@end}[4]{
	\pgfkeys{/quantikz,wires=1,style=,label style=,braces=,steps=1}%
	\edef\background{0}%
	\pgfkeys{/quantikz,#1}%
	\pgfkeysgetvalue{/quantikz/wires}{\quantwires}%
	\pgfkeysgetvalue{/quantikz/row}{\row}%
	\pgfkeysgetvalue{/quantikz/col}{\col}%
	\pgfkeysgetvalue{/quantikz/steps}{\steps}%
	\edef\fit{(\tikzcdmatrixname-col\col.west |- \tikzcdmatrixname-row\row.north)(\tikzcdmatrixname-col\the\numexpr\col+\steps-1\relax.east |- \tikzcdmatrixname-row\the\numexpr\row+\quantwires-1\relax.south)}%
	\ifthenelse{\background=1}{%
		\begin{scope}[on background layer]\node (ggroup-\row-\col) [fit=\fit,ggroup,label={[group label,#3]:#4},#2] {};\end{scope}
	}{%
		\node (ggroup-\row-\col) [fit=\fit,ggroup,label={[group label,#3]:#4},#2] {};
	}
}

\newcommand{\cwbend}[1]{
	\vcw{#1}\cw
	\edef\cell{\the\pgfmatrixcurrentrow-\the\pgfmatrixcurrentcolumn}
	\expandafter\pgfutil@g@addto@macro\expandafter\tikzcd@atendlabels\expandafter{%
		\expandafter\latephase@end\expandafter{\cell}
	}
}

\newcommand{\latephase@end}[1]{
	\node [phase,inner sep=2pt] at (\tikzcdmatrixname-#1) {};
}

\patchcmd\tikzcd@{\tikzpicture}{\def\toslice{0}\def\vert{0}
 \begin{tikzpicture}}{}{}
\patchcmd\tikzcd@
    {\global\let\tikzcd@savedpaths\pgfutil@empty}
    {\global\let\tikzcd@savedpaths\pgfutil@empty
    \global\let\tikzcd@atendsavedpaths\pgfutil@empty
    \global\let\tikzcd@atendlabels\pgfutil@empty
    \global\let\tikzcd@atendslices\pgfutil@empty
    \global\let\tikzcd@atendglobals\pgfutil@empty
    \pgfutil@g@addto@macro\tikzcd@savedpaths\DivideRowsCols
    \ifthenelse{\toslice=1}{\ifthenelse{\vert=0}{\pgfutil@g@addto@macro\tikzcd@atendslices\sliceallr}{\pgfutil@g@addto@macro\tikzcd@atendslices\sliceallvr}}{}
    }{}{}

\def\endtikzcd{%
  \pgfmatrixendrow\egroup%
  \pgfextra{\global\let\tikzcdmatrixname\tikzlastnode};%
  \tikzcdset{\the\pgfmatrixcurrentrow-row diagram/.try}%
  \begingroup%
    \pgfkeys{
      /handlers/first char syntax/the character "/.initial=\tikzcd@forward@quotes,%
      /tikz/edge quotes mean={%
        edge node={node [execute at begin node=\iftikzcd@mathmode$\fi,
                         execute at end node=\iftikzcd@mathmode$\fi,
                         /tikz/commutative diagrams/.cd,every label,##2]{##1}}}}%
    \let\tikzcd@errmessage\errmessage
    \def\errmessage##1{\tikzcd@errmessage{##1^^J...^^Jl.\tikzcd@lineno\space%
        I think the culprit is a tikzcd arrow in cell \tikzcd@currentrow-\tikzcd@currentcolumn}}%
    \tikzcd@before@paths@hook%
    \tikzcd@savedpaths\tikzcd@atendsavedpaths\tikzcd@atendlabels\tikzcd@atendslices{\globaldefs=1\tikzcd@atendglobals\globaldefs=0}
  \endgroup%
  \end{tikzpicture}%
  \ifnum0=`{}\fi}

\newcommand{\DivideRowsCols}{
	\foreach \n in {1,...,\the\pgfmatrixcurrentrow} {
	\xdef\LoopRow{}
		\foreach \m in {1,...,\the\pgfmatrixcurrentcolumn}{
			\ifnodedefined{\tikzcdmatrixname-\n-\m}{
				\xdef\LoopRow{\LoopRow(\tikzcdmatrixname-\n-\m)}
			}{}
			\ifnodedefined{group\tikzcdmatrixname-\n-\m}{
				\xdef\LoopRow{\LoopRow(group\tikzcdmatrixname-\n-\m)}
			}{}
	}
	\node (\tikzcdmatrixname-row\n) [fit=\LoopRow] {};
	}
	\foreach \n in {1,...,\the\pgfmatrixcurrentcolumn} {
	\xdef\LoopCol{}
		\foreach \m in {1,...,\the\pgfmatrixcurrentrow}{
			\ifnodedefined{\tikzcdmatrixname-\m-\n}{
				\xdef\LoopCol{\LoopCol(\tikzcdmatrixname-\m-\n)}
			}{}
			\ifnodedefined{group\tikzcdmatrixname-\m-\n}{
				\xdef\LoopCol{\LoopCol(group\tikzcdmatrixname-\m-\n)}
			}{}
	}
	\node (\tikzcdmatrixname-col\n) [fit=\LoopCol] {};
	}
}

\providecommand{\setmiddle}[1]{%
\IfInteger{#1}{
\pgfmathtruncatemacro\wholepart{floor(#1)}
\edef\temp{\noexpand\tikzset{%
/tikz/baseline={([yshift=-\MathAxis]\noexpand\tikzcdmatrixname-\wholepart-1.base)}
}}
\temp
}{%
\pgfmathtruncatemacro\wholepart{floor(#1)}
\pgfmathtruncatemacro\neighbour{floor(#1)+1}
\pgfmathsetmacro\fractionalpart{#1-floor(#1)}
\edef\temp{\noexpand\tikzset{%
/tikz/baseline={([yshift=-\MathAxis]$(\noexpand\tikzcdmatrixname-\wholepart-1.base)!\fractionalpart!(\noexpand\tikzcdmatrixname-\neighbour-1.base)$)}
}}
\temp
}
}

\pgfkeys{/tikz/slice all/.code={\def\toslice{1}},/tikz/remove end slices/.initial=0,/tikz/slice titles/.initial={\col},/tikz/slice style/.initial={},/tikz/slice label style/.initial={},/tikz/thin lines/.code={\resetstyles},/tikz/transparent/.code={\maketransparent},/tikz/vertical slice labels/.code={\def\vert{1}},/tikz/align equals at/.code={\setmiddle{#1}}}
\pgfkeys{/quantikz/.is family,/quantikz,%
.unknown/.style={%
	/quantikz/wires=\pgfkeyscurrentname
},%
wires/.initial=1,%
style/.initial={},label style/.initial={},braces/.initial={},background/.code={\def\background{1}},alternate/.default=1,alternate/.code={\def\helper{#1}},row/.initial=1,col/.initial=1,steps/.initial=1,Strike Width/.initial=0.08cm,Strike Height/.initial=0.12cm,swap/.code={\def\toswap{1}},disable auto height/.code={\def\DisableMinSize{1}},cwires/.initial={-1},nwires/.initial={-1},bundle/.initial={-1},brackets/.is choice,brackets/.initial=both}
\pgfkeys{/quantikz,brackets/.cd,none/.code={\def\leftbrace{0}\def\rightbrace{0}},left/.code={\def\leftbrace{1}\def\rightbrace{0}},right/.code={\def\leftbrace{0}\def\rightbrace{1}},both/.code={\def\leftbrace{1}\def\rightbrace{1}}}

\providecommand{\ket}[1]{\ensuremath{\left|#1\right\rangle}}
\providecommand{\bra}[1]{\ensuremath{\left\langle#1\right |}}

\providecommand{\braket}[2]{\ensuremath{\left\langle#1\middle|#2\right\rangle}}

\newcommand{\push}[1]{#1 \qw}

\DeclareExpandableDocumentCommand{\phase}{O{}m}{|[phase,#1,label={[phase label,#1]#2}]| {} \qw}
\DeclareExpandableDocumentCommand{\control}{O{}m}{|[phase,#1]| {} \qw}
\DeclareExpandableDocumentCommand{\ocontrol}{O{}m}{|[ophase,#1]| {} \qw}
\DeclareExpandableDocumentCommand{\targ}{O{}m}{|[circlewc,#1]| {} \qw}
\DeclareExpandableDocumentCommand{\targX}{O{}m}{|[crossx2,#1]| {} \qw}

\DeclareExpandableDocumentCommand{\meter}{O{}{m}}{|[meter,label={[my label]#2},#1]| {} \qw}
\DeclareExpandableDocumentCommand{\measuretab}{O{}{m}}{|[measuretab,#1]| {#2} \qw}
\DeclareExpandableDocumentCommand{\meterD}{O{}{m}}{|[meterD,#1]| {#2} \qw}
\DeclareExpandableDocumentCommand{\measure}{O{}{m}}{|[measure,#1]| {#2} \qw}
\DeclareExpandableDocumentCommand{\trash}{O{}{m}}{|[trash,label={below:#2},#1]| {} \qw}

\DeclareExpandableDocumentCommand{\ctrlbundle}{O{1}O{}m}{|[phase bundle,#2]| {} \vqw{#3}\qwbundle[alternate=#1]{}}

\def\swap#1{%
	\targX{}
	\edef\start{\the\pgfmatrixcurrentrow-\the\pgfmatrixcurrentcolumn}
	\edef\end{\the\numexpr#1+\pgfmatrixcurrentrow\relax-\the\pgfmatrixcurrentcolumn}
	\expandafter\expandafter\expandafter\vqwexplicitcenter\expandafter\expandafter\expandafter{\expandafter\start\expandafter}\expandafter{\end}
}

\newcommand{\vcw}[1]{
	\edef\start{\the\pgfmatrixcurrentrow-\the\pgfmatrixcurrentcolumn}
	\edef\end{\the\numexpr#1+\pgfmatrixcurrentrow\relax-\the\pgfmatrixcurrentcolumn}
	\expandafter\expandafter\expandafter\vcwexplicit\expandafter\expandafter\expandafter{\expandafter\start\expandafter}\expandafter{\end}
}
\newcommand{\vqw}[1]{
	\edef\start{\the\pgfmatrixcurrentrow-\the\pgfmatrixcurrentcolumn}
	\edef\end{\the\numexpr#1+\pgfmatrixcurrentrow\relax-\the\pgfmatrixcurrentcolumn}
	\expandafter\expandafter\expandafter\vqwexplicit\expandafter\expandafter\expandafter{\expandafter\start\expandafter}\expandafter{\end}
}
\newcommand{\vqwexplicit}[2]{
	\arrow[from=#1,to=#2,arrows] {}
}
\newcommand{\vqbundleexplicit}[2]{
	\arrow[from=#1,to=#2,arrows] {} \arrow[from=#1,to=#2,arrows,yshift=0.1cm] {}\arrow[from=#1,to=#2,arrows,yshift=-0.1cm] {}
}
\newcommand{\vcwexplicit}[2]{
	\arrow[from=#1,to=#2,arrows,xshift=0.05cm] {}\arrow[from=#1,to=#2,arrows,xshift=-0.05cm] {}
}
\newcommand{\vcwhexplicit}[2]{
	\arrow[from=#1,to=#2,arrows,yshift=0.05cm] {}\arrow[from=#1,to=#2,arrows,yshift=-0.05cm] {}
}
\newcommand{\vqwexplicitcenter}[2]{
	\arrow[from=#1,to=#2,arrows,start anchor=center,end anchor=center] {}
}

\newcommand{\qw}{\ifthenelse{\the\pgfmatrixcurrentcolumn>1}{\arrow[arrows]{l}}{}}
\newcommand{\cw}{\ifthenelse{\the\pgfmatrixcurrentcolumn>1}{\arrow[arrows,yshift=0.05cm]{l}\arrow[arrows,yshift=-0.05cm]{l}}{}}
\newcommand{\qwbundle}[2][]{\ifthenelse{\the\pgfmatrixcurrentcolumn>1}{
	\def\helper{0}
	\pgfset{/quantikz,#1}
	\ifthenelse{\helper>0}{
		\arrow[arrows,yshift=0.1cm]{l}
		\ifthenelse{\helper=1}{\arrow[arrows]{l}}{}
		\arrow[arrows,yshift=-0.1cm]{l}
	}{
	\arrow[phantom,strike arrow]{l}[xshift=\pgfkeysvalueof{/quantikz/Strike Width}, yshift=\pgfkeysvalueof{/quantikz/Strike Height},anchor=south west,inner sep=0pt]{\scriptstyle #2}\qw
	}}{}}

\tikzcdset{row sep/normal=0.5cm,column sep/normal=0.5cm,thick,nodes in empty cells,
	every cell/.style={
        anchor=center,minimum size=0pt,inner sep=0pt,outer sep=0pt,thick
    },
    arrows/.style={dash,thick},
    1-row diagram/.style={}
    }
\tikzset{
	thickness/.style={thick},
    operator/.style={draw,fill=white,minimum size=1.5em, inner sep=2pt,thickness,align=center},
    ggroup/.style={draw,minimum size=1.5em,thickness,align=center,inner sep=4pt},
    leftinternal/.style={anchor=mid west,font=\scriptsize,inner sep=4pt,align=center},
    rightinternal/.style={anchor=mid east,font=\scriptsize,inner sep=4pt,align=center},
    wave/.style={inner sep=-3pt,tape,fill=white,apply={draw=black} except on segments {5,6,1,2,9}},
    phase/.style={fill,shape=circle,minimum size=4pt},
    phase bundle/.style={fill,shape=rectangle,rounded corners=1.5pt,minimum width=4pt,minimum height=10pt},
    phase label/.style={label distance=2mm,anchor=mid,label position=45},
    ophase/.style={fill=white,draw=black,shape=circle,minimum size=4pt},
    internal/.style={thickness,black},
    line/.style={path picture={ 
\draw[internal](path picture bounding box.west) -- (path picture bounding box.east);
}},
	linecont/.style={circle,line},
    cross/.style={path picture={ 
\draw[internal](path picture bounding box.north) -- (path picture bounding box.south) (path picture bounding box.west) -- (path picture bounding box.east);
}},
    circlewc/.style={draw,circle,cross,minimum width=4pt,inner sep=3pt},
    crossx/.style={path picture={ 
\draw[internal,inner sep=0pt]
(path picture bounding box.south east) -- (path picture bounding box.north west) (path picture bounding box.south west) -- (path picture bounding box.north east) (path picture bounding box.west) -- (path picture bounding box.east);
}},
	crossx2/.style={circle,crossx,minimum size=1em},
	trash/.style={path picture={\draw[internal,inner sep=0pt,-stealth] (path picture bounding box.west) -- (path picture bounding box.center) -- (path picture bounding box.south);},minimum height=2.5em,minimum width=2em},
    dd/.style={decoration={brace},decorate,thickness},
    dm/.style={decoration={brace,mirror},decorate,thickness},
    slice/.style={thickness,red,dash pattern=on 5pt off 3pt,align=center},
    meter/.style={draw,fill=white,minimum width=2em,minimum height=1.5em, rectangle, font=\vphantom{A}, thickness,
 path picture={\draw ([shift={(.1,.24)}]path picture bounding box.south west) to[bend left=50] ([shift={(-.1,.24)}]path picture bounding box.south east);\draw[-{Latex[scale=0.6]}] ([shift={(0,.1)}]path picture bounding box.south) -- ([shift={(.3,-.1)}]path picture bounding box.north);}},
 	measuretab/.style={draw,signal,signal to=west,inner sep=4pt,fill=white},
 	meterD/.style={draw,rounded rectangle,rounded rectangle left arc=none,inner sep=4pt,fill=white},
 	measure/.style={draw,rounded rectangle,inner sep=4pt,fill=white},
 	my label/.style={yshift=0.1cm,above,align=center},
 	gg label/.style={label position=center,align=center},
 	group label/.style={label position=above,yshift=0.2cm,anchor=mid},
 	strike arrow/.style={
    decoration={markings, mark=at position 0.5 with {
        \draw [internal,-] 
            ++ (-\pgfkeysvalueof{/quantikz/Strike Width},-\pgfkeysvalueof{/quantikz/Strike Height} )
            -- ( \pgfkeysvalueof{/quantikz/Strike Width}, \pgfkeysvalueof{/quantikz/Strike Height});}	
    },
    postaction={decorate},
}
}

\def\resetstyles{
\tikzcdset{thin,every cell/.append style={thin},arrows/.append style={thin}
    }
\tikzset{
	thickness/.style={thin},
	meter/.append style={thin},
	phase/.append style={minimum size=3pt},
	ophase/.append style={minimum size=3pt},
}
}

\def\maketransparent{
\tikzset{
	operator/.append style={fill opacity=0},
	meter/.append style={fill opacity=0}
}
} \makeatother

\usepackage[unicode,colorlinks]{hyperref}

\usepackage{cleveref}
\crefname{gpair}{pair}{pairs}
\creflabelformat{gpair}{#2(#1)#3}
\crefname{prop}{property}{properties}
\crefname{figure}{\figurename}{\figurename}
\crefname{table}{\tablename}{\tablename}
\Crefname{equation}{Eq.}{Eqs.}

\newcommand{\eqdef}{\coloneqq} 
\newcommand{\integersbelow}[1]{\mathbb{Z}_{#1}}
\newcommand{\tra}{\mathsf{T}}

\newcommand{\fpi}[2][]{\frac{#1\pi}{#2}}
\newcommand{\lpi}[2][]{#1\mkern1mu\pi/#2}

\newcommand{\perr}{p_{\text{err}}}

\usepackage{tikz}
\tikzset{%
    vert/.style={circle, fill=black, scale=0.18},
    edge/.style={thick},
}

\newcommand{\mhgraph}[4][]{
	\begin{tikzpicture}[#1]
    \def\colors{"blue","orange","green!50!black","purple"};
    \def\r{1}; \def\shift{1.3mm}; \def\d{#2};
	
	\path 
	    \foreach \i in {1,2,...,2\d} {
	        -- (0,0) -- ++({\i*180/\d}:\r) node[vert]{} coordinate (v-\i)
            \foreach \x/\sh/\edges in {in/-\shift/{#3},out/\shift/{#4}} {
                -- (0,0) -- ++({\i*180/\d}:\r)
                \foreach [count=\j] \edge in \edges {
    	            -- ++({\i*180/\d}:\sh) coordinate (e-\x-\j-\i)
    	        }
            }
	    };
	
    \foreach \x/\edges in {in/{#3},out/{#4}} {
        \foreach \edge/\p [count=\c,evaluate=\c as \col %
                           using {{\colors}[\c-1]}] in \edges {
            \draw[edge,\col]
                \foreach [count=\j] \vert in \edge {
                    \ifnum \j > 1  --  \fi
                    (e-\x-\p-\vert)
                };
        }
    }
    
	\end{tikzpicture}
}

\newcommand{\topdf}{\texorpdfstring}

\begin{document}

\title{Breaking simple quantum position verification protocols with little entanglement}

\author{Andrea Olivo}
\email{andrea.olivo@inria.fr}
\affiliation{Inria, Paris, France}
\affiliation{LPGP, CNRS, Université Paris-Saclay, 91405 Orsay Cedex, France}

\author{Ulysse Chabaud}
\email{ulysse.chabaud@lip6.fr}
\affiliation{Sorbonne Université, CNRS, LIP6, 4 place Jussieu, F-75005 Paris, France}

\author{André Chailloux}
\email{achaillo@inria.fr}
\affiliation{Inria, Paris, France}

\author{Frédéric Grosshans}
\email{frederic.grosshans@lip6.fr}
\affiliation{Sorbonne Université, CNRS, LIP6, 4 place Jussieu, F-75005 Paris, France}

\begin{abstract}
Instantaneous nonlocal quantum computation (INQC) evades apparent quantum and relativistic
constraints and allows to attack generic quantum position verification (QPV) protocols---%
aiming at securely certifying the location of a distant prover---at an exponential
entanglement cost.
We consider adversaries sharing maximally entangled pairs of
qudits and find low-dimensional INQC attacks against the simple practical
family of QPV protocols based on single photons polarized at an angle $\theta$.
We find exact attacks against some rational angles,
including some sitting outside of the Clifford hierarchy (e.g.\ $\lpi{6}$), and show
no $\theta$ allows to tolerate errors higher than $\simeq 5\cdot 10^{-3}$ against
adversaries holding two ebits per protocol's qubit.
\end{abstract}

\maketitle

\section{Introduction} \label{sec:intro}

The interplay between
quantum constraints on measurements and relativistic effects is very subtle,
as witnessed, among others, by the famous Bohr--Einstein debate \cite{Bohr49}.
As early as 1931, Landau and Peierls \cite{LandauPeierls31} showed the measurement
of the electromagnetic field at a specific location to be nonlocal
and therefore deduced its impossibility.
However, in 1980 Aharonov and Albert \cite{AharonovAlbert1980} started a line 
of research investigating how to harness entanglement
to perform non-local measurements and operations without violating
causality---a feat now called \emph{instantaneous nonlocal quantum computation} (INQC);
e.g.\ they showed in 1981 \cite{AharonovAlbert1981} how to perform what we now
call a Bell measurement between two distant particles using an entangled pair
of qubits. These results were generalized to other observables
\cite{AharonovAlbert1984I,AharonovAlbert1984II,AAV86,PopescuVaidman94,GroismanVaidman01},
until Vaidman showed in 2003 \cite{Vaidman03}
how to approximate any nonlocal measurement
using teleportation \cite{bennett1993teleporting} and
causal classical communications.

In 2009 Chandran et al.~\cite{CGMO09} investigated a cryptographic primitive
known as position verification (namely, the task of
certifying a prover is at a specific location) exploiting timing constraints and
the relativistic speed limit on information propagation.
Applications include encryption decipherable only at a specific distant
location.
They showed this primitive to be insecure
in the classical setting, even under computational assumptions:
a coalition of colluding adversaries mimicking the honest prover's actions
can break any classical protocol by copying and sharing the data sent by the verifiers.
The following year,\footnote{We refer the reader interested in the involved history
  of the first QPV protocol to \cite{SchaffnerPBQCWeb}}
three groups independently proposed quantum position verification (QPV) protocols
\cite{PatentKMSB06,KMS11,CFG+10,Malaney10a,Malaney10b,BCF+14}, building on the no-cloning
properties of quantum mechanics.

However, because of its universality \cite{Vaidman03}, INQC turned out to be a powerful tool
to attack QPV protocols \cite{BCF+14};
their security is not unconditional, but is characterized by the resources
needed to perform the associated INQC protocol.
This prompted new, quantitative investigations into INQC.
Beigi and König used port-based teleportation~\cite{ishizaka2009quantum}
to reduce the entanglement cost of the
universal approximate attack \cite{Vaidman03,BCF+14} from doubly exponential in the number
of used qubits to simply exponential.
The cost has been further reduced for the exact implementation of specific families
of nonlocal unitaries, either in the Clifford group \cite{KMS11,LauLo11,GC20},
as well as operations  finite depth in the Clifford hierarchy \cite{CL15,speelman2015instantaneous},
and teleportation routed according to distributed classical functions~\cite{BFSS13,KlauckPodder14}.
Recently, Gonzales and Chitambar \cite{GC20} improved the implementation
of arbitrary two qubit unitaries.

Security proofs for QPV have proven to be elusive,
with the notable exception of the hash-function
based protocol proposed by Unruh~\cite{Unruh14}, which requires exponentially
many queries in the random oracle model.
The other published results correspond to lower bounds on the amount of
entanglement needed to spoof a QPV protocol by INQC.
Some protocols have initially been proven to require entanglement~\cite{BCF+14},
then security was extended through entropic reasoning to lower bounds smaller than one
entangled pair per qubit~\cite{TFKW13}.
While the improved bound in~\cite{RG15} is tight
for a simple protocol~\cite{KMS11},
it is still exponentially far from the best known universal attack.

In this letter we focus on QPV$_\theta$, one of the simplest protocol classes for QPV.
This family of protocols is a straightforward generalization of one of the first QPV
protocols, which was inspired by Bennett and Brassard’s seminal
quantum key distribution (QKD) protocol~\cite{BB84},
and has been used through most of the literature on QPV.
On a practical side, QPV$_\theta$ would need a relatively simple set-up,
similar to what is currently being developed for free-space QKD
\cite{Pugh+17,Liao+17,Avesani+19} with stringent timing constraints.
Most protocols in QPV$_\theta$ are resistant to exact attacks from adversaries
pre-sharing a maximally entangled pair of qubits or qutrits~\cite{LauLo11}.
It is therefore natural to ask how well can attackers do on these near-term
protocols with entangled states of bigger (but still of practical interest) dimension.
On the theoretical side, it will hopefully help to gain
better insights into INQC:
tuning the single parameter $\theta$ allows us to explore the Clifford hierarchy,
greatly changing the cost of known exact attacks \cite{KMS11,CL15,GC20}.
More generally, known attacks to various QPV protocols seem to hint towards a profound
link with deeper questions in quantum information, from generalized teleportation 
schemes \cite{BeigiKoenig11,BFSS13,KlauckPodder14}
to quantum compilation of INQC unitaries
\cite{CL15,speelman2015instantaneous,GC20}.

Here, we study INQC protocols using a small amount of entanglement through
the attacks against QPV$_\theta$, presented in \cref{sec:protocol} along with the attack model.
We describe, in \cref{sec:circuitrep}, a new circuit representation of the QPV$_\theta$ protocols
and use it in \cref{sec:exactattacks} to characterize exact attacks exploiting entangled qudit
pairs of dimension $d\le12$, finding the most efficient INQC protocols to date
for many angles, including some out of the Clifford hierarchy.
In \cref{sec:approxattacks}, we use it to numerically explore the best
approximate attacks for $d\le5$.

\section{Protocol and attack model}
    \label{sec:protocol}

\begin{figure}[b]
    \centering
    \captionsetup[subfloat]{farskip=0pt}
    \subfloat[honest prover at $P$\label{fig:diagQPVhonest}]{%
    \begin{tikzpicture}[scale=1.3]
        \draw [->] (0,0) -- (0,3.25) coordinate (yaxis) node[left] {$t$};
        \draw (0,0) -- (2.8,0) coordinate (xaxis);
        \draw decorate[decoration={snake,amplitude=0.6mm}]
            {(0.2,0) coordinate (a_1) -- ++(45:2.1)} -- ++(45:1.15)  coordinate (a_2);
        \draw (2.5,0.7) coordinate (b_1) -- ++(135:3.25) coordinate (b_2);
        \coordinate (p) at (intersection of a_1--a_2 and b_2--b_1);
        \fill[red!90!black] (p) circle (1.5pt);
        \path (a_1) node[below] {$V_1$} ++(48:1.7) node[left] {$(R_\theta)^b\ket{x}$}
            (b_1) ++(140:0.4) node[above right] {$b$} (b_1 |- 0,0) node[below] {$V_2$}
            (p |- 0,0) node[below,red!90!black] {$P$} (p) ++(130:0.5) node[above] {$x$}
            (p) ++(50:0.5) node[above] {$x$};
        \draw[gray] (a_1) -- ++(90:3.1) (b_1 |- 0,0) -- ++(90:3.1);
        \draw[dashed,gray] (yaxis |- a_1) ++(0,0.03) node[left] {$t_1$} -- ++(0:2.7)
            (yaxis |- b_1) node[left] {$t_2$} -- ++(0:2.7)
            (yaxis |- b_2) node[left] {$t_4$} -- ++(0:2.7)
            (yaxis |- a_2) node[left] {$t_3$} -- ++(0:2.7);
    \end{tikzpicture} }
    \hfill
    \subfloat[colluding attackers at $A,B$\label{fig:diagQPVattack}]{%
    \begin{tikzpicture}[scale=1.3]
        \draw [->] (0,0) -- (0,3.25) coordinate (yaxis) node[left] {$t$};
        \draw (0,0) -- (2.8,0) coordinate (xaxis);
        \draw decorate[decoration={snake,amplitude=0.6mm}]
            {(0.2,0) coordinate (a_1) -- ++(45:1.35) coordinate (e_1)}
            ++(45:1.25) coordinate (ee_2) -- ++(45:0.65) coordinate (a_2);
        \path (2.5,0.7) coordinate (b_1)  (ee_2 |- 0,0) coordinate (temp_2)
            (b_1) ++(135:2) coordinate (temp_3)  (e_1 |- 0,0) coordinate (temp_1)
            (intersection of ee_2--temp_2 and b_1--temp_3) coordinate (e_2)
            (intersection of e_1--temp_1 and b_1--temp_3) coordinate (ee_1)
            (e_1) ++(-45:1) coordinate (temp_4)  (e_2) ++(-135:1) coordinate (temp_5)
            (intersection of e_1--temp_4 and e_2--temp_5) coordinate (h);
        \draw (b_1) -- (e_2) (ee_1) -- ++(135:1.35) coordinate (b_2);
        \draw[blue,dashed] (e_1) ++(41:0.4) node[below] {$u$}
            (ee_1) ++(-38:0.5) node[above] {$b,s$} (e_1) -- (ee_1) node[left,midway] {$u$}
            -- (e_2) -- (ee_2) node[right,midway] {$b,s$} -- (e_1);
        \draw[blue] decorate[decoration={snake,amplitude=0.6mm},segment length=5.5]
            {(h) -- (e_1) (h) -- (e_2)} (h) node[below] {$\ket{\Phi}$};
        \foreach \x in {e_1,ee_1,e_2,ee_2,h} {\fill[blue] (\x) circle (1pt);}
        \path (a_1) node[below] {$V_1$} ++(48:0.7) node[left] {$\ket{\psi}$}
            (b_1) ++(140:0.4) node[above right] {$b$} (b_1 |- 0,0) node[below] {$V_2$}
            (ee_1) ++(130:0.4) node[above] {$x$} (ee_2) ++(50:0.2) node[above] {$x$};
        \path[blue] (e_1 |- 0,0) node[below] {$A$} (e_2 |- 0,0) node[below] {$B$};
        \draw[gray] (a_1) -- ++(90:3.1) (b_1 |- 0,0) -- ++(90:3.1);
        \draw[dashed,gray] (yaxis |- a_1) ++(0,0.03) node[left] {$t_1$} -- ++(0:2.7)
            (yaxis |- b_1) node[left] {$t_2$} -- ++(0:2.7)
            (yaxis |- b_2) node[left] {$t_4$} -- ++(0:2.7)
            (yaxis |- a_2) node[left] {$t_3$} -- ++(0:2.7);
    \end{tikzpicture} }
    \caption{
        Spacetime diagrams of QPV$_\theta$ protocol and attack model.
        Lines at 45° represent lightspeed quantum (ondulated) and classical
        (straight, solid and dashed) channels.
        \protect\subref{fig:diagQPVhonest}
        When a prover is present at $P$, he measures the quantum input in the
        correct basis and broadcasts the measurement result $x$ back to $V_1$ and $V_2$.
        \protect\subref{fig:diagQPVattack}
        Attackers have access to locations $A$ and $B$ and share a  quantum
        resource $\ket{\Phi}$.
        They share the classical outcomes of their measurements and attempt to
        reconstruct $x$ in time to be broadcast back to the verifier.
    }
    \label{fig:diagQPV}
\end{figure}
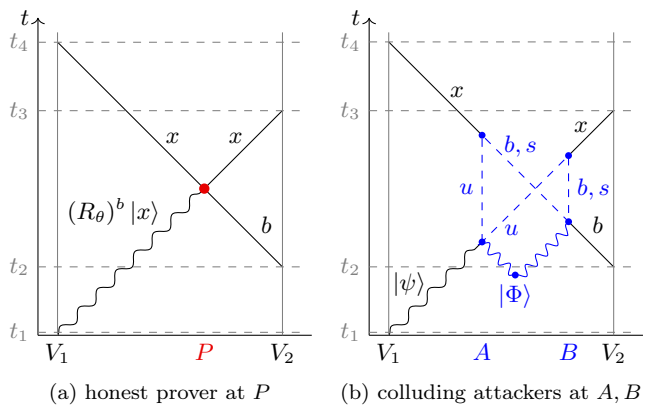

We constrain all parties to 1D space.
We give a more detailed definition and a simple attack in \cref{sec:qpvprot}.
\paragraph*{QPV$_\theta$ protocol}
The verifiers control two stations $V_1$ and $V_2$ to the left
and to the right of the prover's claimed position $P$ (\cref{fig:diagQPVhonest}).
By making use of shared random bits $x,b \in \{0,1\}$, they
prepare the qubit $\ket{\psi} = (R_{\theta})^b \ket{x}$ at $V_1$, where
\begin{align}
    R_\theta = \begin{pmatrix}
                   \cos\theta & -\sin\theta \\
                   \sin\theta & \cos\theta
               \end{pmatrix}
    &&
    \theta \in [0,2\pi].
\end{align}
They then send $\ket{\psi}$ from $V_1$ and $b$ from $V_2$, timed in such a way to arrive
simultaneously at $P$.
The prover carries out a measurement on $\ket{\psi}$ in a suitable basis depending on $b$
and broadcasts the classical measurement result $x$ to $V_1$ and $V_2$,
who check they received the correct bit at the expected time.
After repeating for multiple rounds, they consider the position $P$
authenticated if the prover answered correctly often enough.

\paragraph*{Attack model.}
The choice of the attack model, leading us to new attacks, is inspired by the
teleportation attack for the BB84 protocol~\cite{KMS11,LauLo11} (\cref{sec:telep_attack}).
The attackers Alice and Bob have no access to
the location $P$ to be authenticated, but control two stations $A$ and
$B$ respectively located between $V_1$ and $P$ and between $P$ and $V_2$;
a resource quantum state $\ket{\Phi}$ is pre-shared between the two stations.
Special relativity allows them one round of simultaneous communication.
As the protocol requires them to output a classical message, we constrain internal
communication to be classical as well (LOBC model, see~\cite{GC20});
limited quantum communication can be included through teleportation.
All quantum operations are assumed to be unitary, as we focus mainly on the
dimensionality $d$ of $\ket{\Phi}$:
general CP maps can be extended to unitary operators through a Stinespring
dilation~\cite{wilde_2013} using only local resources.
We choose $\ket{\Phi}$ to be a maximally entangled qu$d$it pair in order to exploit
some of its properties, noting that this choice leads to an optimal attack for
the BB84 protocol~\cite{RG15}.
Finally, Alice and Bob act identically and separately on each round.

This attack model translates to the INQC implementation of a special family of
two-qubit nonlocal unitaries:
\begin{equation}
    U_\theta = \text{CNOT}_\text{AB}
        \left(I \otimes \ket{0}\bra{0} + R_{-\theta} \otimes \ket{1}\bra{1}\right),
\end{equation}
making it easier to compare it to known attacks,
in particular the efficient ones in~\cite{GC20} (see \cref{sec:INQC}).

\section{Circuit picture}
    \label{sec:circuitrep}

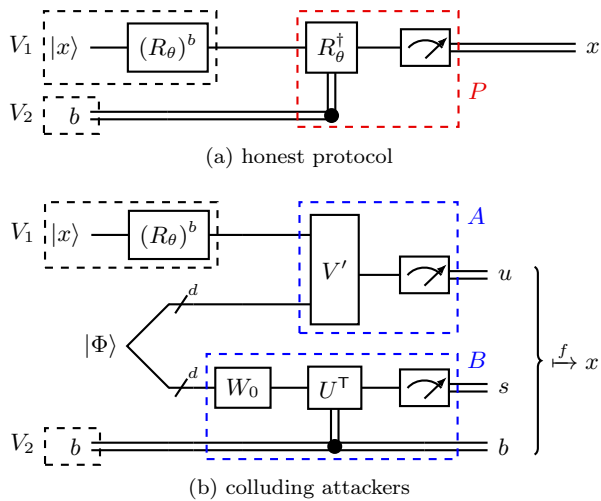
\begin{figure}
    \centering
    \subfloat[honest protocol\label{fig:circQPVhonest}]{%
    \begin{quantikz}
        \gategroup[steps=3, style={dashed, inner ysep=1pt, inner xsep=0pt},
            label style={label position=left, anchor=east, yshift=-.5em}]{$V_1$}
         & \lstick{$\ket{x}$} & \gate{(R_\theta)^b} &[23pt] 
         \gate{R_\theta^\dag}
        \gategroup[2,steps=2,
            style={dashed, draw=red!90!black, inner ysep=1pt, inner xsep=0pt},
            label style={red!90!black, label position=right, anchor=west, yshift=-1.5em}
            ]{$P$}
         &[2pt] \meter{} &[33pt] \cw \rstick{$x$}
        \\[2pt]
        \gategroup[steps=2, style={dashed, inner ysep=0pt, inner xsep=0pt},
            label style={label position=left, anchor=east, yshift=-.5em}]{$V_2$}
        & \lstick{$b$} & \cw & \cwbend{-1} & 
    \end{quantikz}}

    \subfloat[colluding attackers\label{fig:circQPVattack}]{%
    \begin{quantikz}
        \gategroup[steps=3, style={dashed, inner sep=0pt},
            label style={label position=left, anchor=east, yshift=-.5em}]{$V_1$}
        & \lstick{$\ket{x}$} & \gate{(R_\theta)^b} &[-20pt] \qw &[-6pt] \qw
        & \gate[3,nwires={2}]{V'}
        \gategroup[3,steps=2, style={dashed, draw=blue,  inner sep=0pt},
            label style={blue, label position=right,anchor=west,yshift=1.5em}]{$A$} & &
        \\[-18pt]
        & & & & & & \meter{} & \cw \rstick{$u$} & \rstick[4]{$\xmapsto{\smash{f}} x$}
        \\[-18pt]
        & & \makeebit{$\ket{\Phi}$} & \qwbundle{d} & \qw & & & &
        \\
        & & & \qwbundle{d} & \gate{W_0}
        \gategroup[2,steps=3,
            style={dashed, draw=blue, inner xsep=0pt, inner ysep=1.5pt},
            label style={blue, label position=right,anchor=west,yshift=1.3em}]{$B$}
        & \gate{U^\tra} & \meter{} & \cw \rstick{$s$} &
        \\
        \gategroup[steps=2, style={dashed, inner sep=0pt},
            label style={label position=left, anchor=east, yshift=-.5em}]{$V_2$}
        & \lstick{$b$} & \cw & \cw & \cw & \cwbend{-1} & \cw & \cw \rstick{$b$} &
    \end{quantikz} }
    \caption{
        Circuit representation of the spacetime diagrams in \cref{fig:diagQPV},
        where the actions of the verifier, prover and attackers correspond here
        to the dashed boxes.
        The causal relations are enforced by the wires between the boxes;
        the final broadcasting of $x$ is not represented.
    }
    \label{fig:circQPV}
\end{figure}

A precise representation of QPV$_\theta$, both with honest and
cheating provers, is described by the spacetime circuits~\cite{Unruh14}
of \cref{fig:circQPV}.
Alice and Bob's strategy consists in obtaining (clonable) classical information by
interacting their respective inputs with local resources, from which they have
to deduce $x$.
Alice ignores the basis $b$ in which the incoming qubit $(R_\theta)^b \ket{x}$
is encoded, and her actions are modeled by a unitary operation $V'$ acting on
both the verifier's qubit and her half of the entangled qudit pair, followed by
a measurement in the computational basis.
She forwards her outcome $u\in\integersbelow{2d}$ to Bob.

Bob knows the basis $b\in\{0,1\}$ but has only access to his half of the qudit pair, to
which he applies a unitary $W_b$ followed by a measurement in the computational basis.
He obtains $s\in\integersbelow{d}$ he forwards to Alice along with $b$.
Without loss of generality, we define $U^\tra \eqdef W_1^{\vphantom{\dag}} W_0^\dag$,
with $W_b = (U^\tra)^b W_0$;
this allows to rewrite Bob's unitary as a fixed gate $W_0$
followed by a gate $U^\tra$ conditioned on $b$.
The attack is then completed by a classical map $f(b, s, u)$ that they can both 
separately compute after exchanging their measurement results.

Because $\ket{\Phi}$ is maximally entangled we have
\begin{equation}
    (I \otimes W_b^\tra)\ket{\Phi} = (W_b^{\vphantom{\tra}} \otimes I)\ket{\Phi},
\end{equation}
and we derive a formally equivalent circuit for the attack, by transposing $W_0$ and
$(U^\tra)^b$ to Alice's side.
In this version the unitary $W_0^\tra U^b$ is performed by Alice on \emph{her} half
of the entangled state, while Bob immediately measure his half.
Setting $V \eqdef V'(I \otimes W_0^\tra)$, we obtain the reduced circuit
of \cref{fig:circReduction}, in which $\ket{s}$ is the uniformly distributed
computational basis state onto which Alice's qudit is collapsed by Bob's measurement.
While this simplified circuit gives a leaner description of the problem,
it does not preserve the spacetime locality of the operations:
in the real world, Alice has no access to $b$.

\begin{figure}[b]
    \centering
    \begin{quantikz}
        \lstick{$\ket{x}$} &[10pt] \gate{(R_{\theta})^b}
        &[5pt] \gate[3,nwires={2}]{V} &&[-5pt]&&&
        \\[-18pt]
        &&& \qwbundle{2d}\qw & \push{\,\ket{\psi_b(x,s)}\,}
        & \meter{} & \cw \rstick{$u$}
        \\[-18pt]
        \lstick{$\ket{s}$} & \gate{(U)^b}\hphantom{{}_\theta}\qwbundle{d} &&&&&&
    \end{quantikz}
    \caption{
        The reduced circuit.
        It is no longer a spacetime circuit, but is equivalent to 
        \cref{fig:circQPVattack} when $\ket{s}$ is chosen uniformly at random
        and $V \eqdef V'(I\otimes W_0^\tra)$.
        Bob's measurement of his qudit has been omitted.
    }
    \label{fig:circReduction}
\end{figure}
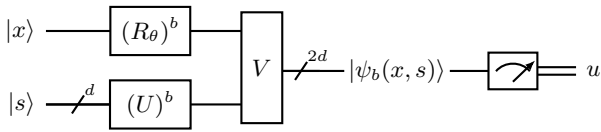

Let $\ket{\psi_b(x,s)}$ be the output state of the reduced circuit before the
measurement,
\begin{equation} \label{eq:psi_bxs}
    \ket{\psi_b(x,s)}\eqdef V(R_{\theta}\otimes U)^b (\ket{x}\otimes\ket{s}).
\end{equation} 
In the following we analyze these states to investigate exact and approximate attacks
against QPV$_\theta$, respectively succeeding with probability $p=1$ and $p<1$.

\section{Exact attacks}
    \label{sec:exactattacks}

Alice and Bob can perform an exact attack if and only if, when given $b$ and $s$,
measuring $\ket{\psi_b(x,s)}$ in the computational basis $\{\ket{u}\}$ is enough
to determine $x$;
namely, the probability of any outcome $u$ has to be zero for at least $x=0$ or $x=1$.
We refer to this requirement as the \textit{deterministic distinguishability condition}:
for all $u\in\integersbelow{2d}$, $s\in\integersbelow{d}$, $b\in\{0,1\}$,
\begin{align} \label{eq:DDC}
    \braket{u|\psi_b(0,s)}=0 \quad \text{or} \quad \braket{u|\psi_b(1,s)}=0,
    && \text{(DDC)}
\end{align}
which is equivalent to $\braket{u|\psi_b(0,s)}\braket{\psi_b(1,s)|u}=0$.
When \cref{eq:DDC} is satisfied, there naturally exists
$f(b,s,u)=x$ for all inputs, giving an exact attack.

\begin{figure}
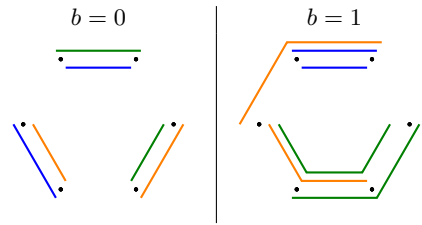

    \begin{tabular}{c|c}
        $b=0$ & $b=1$ \\
        \mhgraph{3}{{1,2}/1,{3,4}/1,{5,6}/1}{{3,4}/1,{5,6}/1,{1,2}/1} &
        \mhgraph{3}{{1,2}/1,{3,4,5}/1,{3,4,5,6}/2}{{1,2}/1,{1,2,3}/2,{4,5,6}/1}
    \end{tabular}
    \caption{A pair of graphs describing an attack for $d=3$}
    \label{fig:hypergraphs}
\end{figure}

Through a custom representation of the output Hilbert space in terms of
hypergraph-inspired objects (\cref{fig:hypergraphs}), we show in \cref{sec:graphproofs}
how to capture some of the restrictions imposed by the DDC.
We use these new tools to analytically characterize exact attacks for $d=2$ and
$d=3$, giving (arguably) simpler proofs for these cases than the ones provided
in~\cite{LauLo11}.
The combinatorial explosion of the above method fundamentally limits its application
to small $d$, even if improvements to $d=4$ might not be totally out of reach.

\begin{table}[b]
    \centering
    \setlength\tabcolsep{4pt}
    \begin{tabular}{c c c c c c c c c c c c}\hline\hline
        $\bm{d}$ &  2  &  3  &  4  &  5  &  6   &  7  &  8  &  9  &  10 &  11 &  12
        \\
        $\bm{k}$ &  4  &  2  &  8  &  4  &8,\,12&  4  &  16 &4,\,6&  20 &  4  &  24
    \\\hline \hline
    \end{tabular}
    \caption{
        New exact attacks for QPV$_\theta$. Depending on the attack dimension
        $d$, we list the values of $k$ for which a valid pair $(U,V)$ breaking
        $\theta=\fpi[n]{k}$ is found $\forall n$.
    }
    \label{tab:attacksfound}
\end{table}

Other approaches are therefore needed: we chose to pursue a numerical method.
From the circuit reduction of \cref{fig:circReduction}, all we need to define an attack
is the pair of unitary matrices $(U,V)$.
Using \cref{eq:psi_bxs}, the DDC may in turn be written as a system of
polynomial equations in the entries of $U$, $V$ and $R_\theta$.
We use a nonlinear least squares method implemented in SciPy~\cite{2020SciPy-NMeth}
to find zeros of the system, as detailed in \cref{sec:num_methods_exact}.
For $d=4$, we quickly find solutions for all angles of the form $\theta=\fpi[n]{8}$,
showing that two ebits are strictly more powerful than an entangled qubit or qutrit.
We then proceed to raise the dimension of the adversaries' entangled qudits up to
$d=12$; we collected our findings in \Cref{tab:attacksfound}.
An interesting pattern emerges: for even $d$, we find an attack for
(at least) all $\theta$ of the form $\fpi[n]{2d}$, and we conjecture
this relation to hold for all even $d$.
Odd dimensions behave differently and appear to be less powerful.
Notably, we find that a pair of maximally entangled six-level systems
is sufficient to break QPV$_{\lpi{6}}$, despite the corresponding rotation
being \emph{outside} of the Clifford hierarchy on qubits.

Direct inspection of the matrices $(U,V)$
have not offered us a straightforward generalization from which an analytic attack
strategy for all $d$ could be derived.
There are a variety of discrete symmetries that are difficult to tackle;
furthermore, our numerical results suggest that the solutions retain some
continuous degrees of freedom.
For some $d$ we present in \cref{sec:explicit_sols} explicit solution matrices,
``reverse-engineered'' from the numerical ones.

\section{Approximate attacks}
    \label{sec:approxattacks}

While the previous method is appropriate to find new exact attacks,
more work is required to gather numerical evidence about the (in)security of
QPV$_\theta$ against adversaries that are allowed a small probability of error.

As detailed in \cref{sec:num_methods_approx}, the error probability for an
attack strategy is:
\begin{equation}\label{eq:p_err}
    \perr = \frac{1}{4d} \sum_{b,s,u}
        \min\big\{ |\braket{u|\psi_b(0,s)}|^2 , |\braket{u|\psi_b(1,s)}|^2 \big\}.
\end{equation}
Minimizing over all attack strategies at fixed $d$,
\begin{align}
    \perr(\theta) = \min_{U,V}\, \perr(U,V,\theta),
\end{align}
we determine an upper bound to the security of QPV$_\theta$.

\begin{figure}[b]
    \centering
    \includegraphics{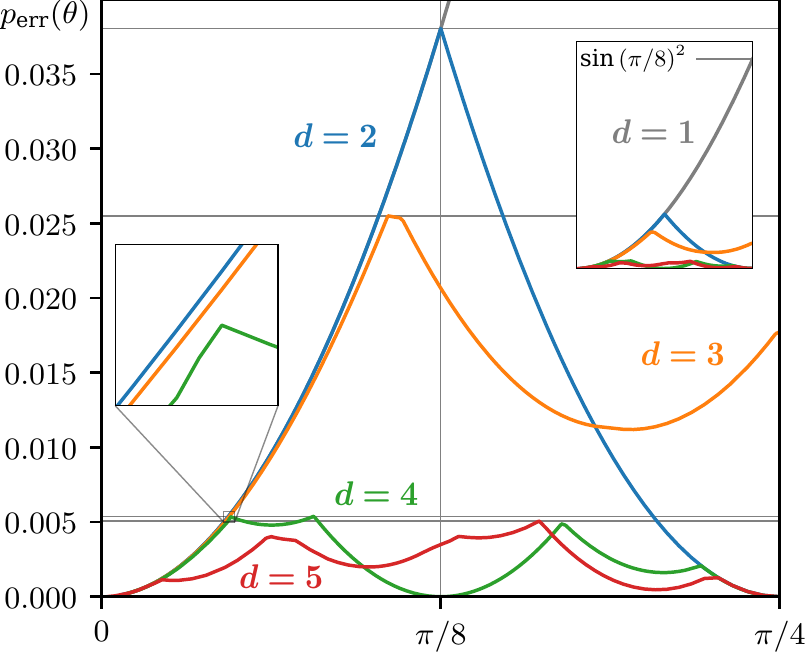}
    \caption{
        The numerically minimized $\perr(\theta)$ for $\theta\in[0,\fpi4]$,
        the other values of $\theta$ being deduced by symmetry.
        Horizontal lines mark the $\max_\theta$ of each curve.
        The $d=1$ analytical curve corresponds to no pre-shared entanglement.
    }
    \label{fig:approx_all}
\end{figure}

Our results (methods in \cref{sec:num_methods_approx})
are plotted in \cref{fig:approx_all} for $d\le5$.
We find a richer structure than what could be expected from the regularity of
the angles of \cref{tab:attacksfound}.
As a nice consistency check, it can be seen that $\perr$ drops to 0 where
we would expect from the exact results.
The shape of the optimal $\perr(\theta)$ curves appears to be of the form
$\min\{p_1(\theta), p_2(\theta), \dots,p_n(\theta)\}$, suggesting that Alice and Bob
may have to employ radically different strategies depending on $\theta$.

\paragraph*{$d=2$}
The adversaries share just one ebit.
The curve found numerically is well fit by
\begin{align}
    \perr(\theta) =
    \begin{cases}
        \sin\left( \frac{\theta}{2} \right)^2 &\, 0\leq\theta\leq\fpi8 \,, \\
        \sin\left( \frac{\theta}{2}-\fpi{8} \right)^2 & \fpi{8}\leq\theta\leq\fpi4 \,.
    \end{cases}
\end{align}
Surprisingly, the probability in the region $0\leq\theta\leq\fpi8$
can be attained without using the ebit at all:
a simple matching strategy is the ubiquitous
\textit{pretty good measurement} (PGM), where Alice can directly measure
the unknown $\ket{\psi} = (R_\theta)^b\ket{x}$ in the intermediate basis
$R_{\theta/2}$ and send the classical result to Bob in the broadcasting phase.

A strategy for the second region ($\fpi{8}\leq\theta\leq\fpi4$) can be obtained
by modifying the the teleportation-based exact attack of \cref{sec:telep_attack},
giving the pair:
\begin{align}
    U = H, && V = (R_{\fpi8 - \frac\theta2} \otimes I) \text{ CNOT}_\text{AB } (H \otimes I).
\end{align}

\paragraph*{$d=3$}
The absence here of an exact attack for ${\theta=\fpi4}$ is more clearly
grasped in the approximate context.
The piecewise function $\perr(\theta)$ seems to involve six curves,
with some strategies prevailing only in small regions of $\theta$
(e.g.\ the ones flattening the ``cusp'' at $\theta/\pi \simeq 0.11$).

\paragraph*{$d=4$}
With two ebits we count five distinct regions, four of which
fit to an expression of the type:
\begin{equation}
    (1 - t) \sin\left(\frac{\theta}{2} - \phi\right)^2 + \frac{t}{2}.
\end{equation}
Around $\theta\simeq\fpi8$ and $\theta\simeq\fpi4$, where $\perr(\theta)$ crosses
the $x$ axis, we have $t=0$ and respectively $\phi=\fpi{16}$, $\phi=\fpi8$.
For both $d=3$ and $d=4$ we find attacks slightly beating the non-entangled PGM in the
region around $\theta\simeq0$; however, for this piece we could not find a
simple analytical formula.

\paragraph*{Multiple bases}
An interesting scenario is QPV$_{(n)}$, an extension of QPV$_\theta$ defined as such:
the verifiers now choose a basis $R_{\theta_b}$, where
$\theta_b$ is picked uniformly from the set
$S_n = \{\lpi[b]{2n}, \forall b \in \integersbelow{n}\}$, and send
$\ket{\psi} = R_{\theta_b}\ket{x}$ from $V_1$ and $b\in\integersbelow{n}$ from $V_2$.
The set $S_n$ is composed of $n$ equally-spaced angles in the range $[0,\fpi2)$;
for $n=2$, QPV$_{(n)}$ reduces to QPV$_{\lpi4}$.
The intuition behind the modified protocol (similar to one suggested in~\cite{KMS11})
is that only Bob can adapt his unitary $U_b$ depending on $b$, making
the constraints on Alice's $V$ tighter and tighter by increasing $n$.
Moreover, QPV$_{(n)}$ allows us to go beyond a direct application of the efficient
attacks devised in~\cite{GC20}.
The numerical optimization (\cref{fig:multibase}) indeed suggests a higher
$\perr(n)$ for large $n$ than the $\max_\theta \perr(\theta)$ for QPV$_\theta$.
The no-entanglement ($d=1$) $\perr(n)$ can be obtained by minimizing the sum of the
squares of the overlaps between a measurement angle $\tilde{\theta}$ and the
angles in $S_n$, giving:
\begin{equation}
    \min_{\tilde{\theta}} \bigg\{
        \frac1n \sum_{\theta_b \in S_n} \sin\left( \tilde{\theta} + \theta_b \right)^2
    \mkern-3mu \bigg\}
    = \frac12 \left[ 1 - \frac1n \csc \left( \frac{\pi}{2n} \right) \right].
\end{equation}

\begin{figure}
    \centering
    \includegraphics{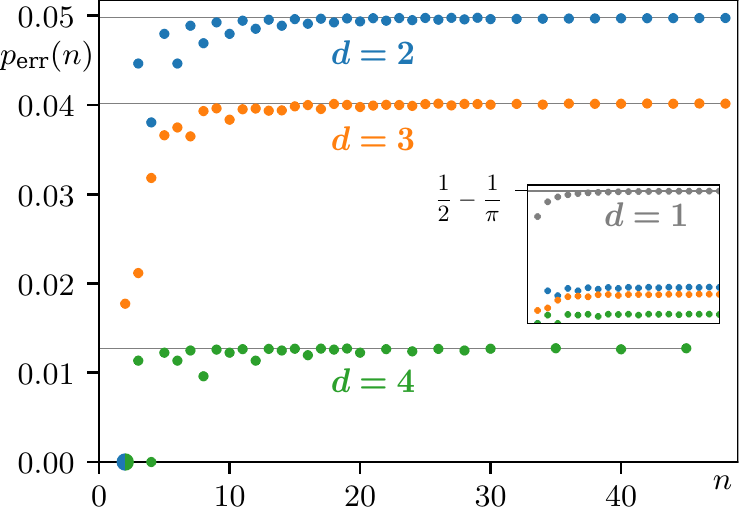}
    \caption{
        The numerically minimized $\perr(n)$ representing attacks to the
        QPV$_{(n)}$ protocol.
        For large $n$ the best attacks found are weaker than the
        ones for QPV$_\theta$.
    }
    \label{fig:multibase}
\end{figure}

\section{Conclusion}
    \label{sec:conclusion}

The family of protocols QPV$_\theta$ explored in this paper holds great promises
for a near-term implementation, due to its experimental and theoretical simplicity.
However, by exploiting its specific structure, we show how adversaries
manipulating a small amount of entanglement are able to perfectly break many angles,
finding new exact attacks and lowering the cost of previously known ones.
We find evidence about the existence of exact attacks for
$\theta$ multiples of~$\fpi{2d}$, and numerically obtain
them for $d\leq12$.
For comparison, attacks in~\cite{GC20} applied to QPV$_\theta$ (the best to date,
as far as we know) require $4n + 15$ ebits to break $\theta=\lpi{2^n}$;
our results suggest INQC protocols consuming just $n-1$ ebits for the corresponding
family of two-qubit nonlocal unitaries $U_\theta$.
Through numerical optimization of approximate attacks, we show that
adversaries manipulating two ebits can attain error probabilities as low as
$\perr \lesssim 5\cdot 10^{-3}$ through the entire $\theta$ range.
 
Some questions about QPV$_\theta$ are left open.
It would be interesting to find an explicit strategy reproducing our attacks
for all $d$.
These results could be useful in other areas,
e.g.\ for designing better gate teleportation protocols.
Moreover, we note that a variant of QPV$_{(n)}$ where the bases are chosen from
the entire Bloch sphere has interesting loss-tolerance properties~\cite{schaffpriv},
and could be a better choice for near-term implementations.
We leave this to future work.

\begin{acknowledgments}
AO et AC acknowledge financial support from ANR project ANR-16-CE39-0001 DEREC.
We also thank Alastair Key for his \LaTeX{} package
\texttt{quantikz}~\cite{kay2018tutorial} we used for the circuits in 
\cref{fig:circQPV,fig:circReduction}.
\end{acknowledgments}

\bibliographystyle{res/apsrevmod4-1}
\bibliography{qpv.bib}

\begin{thebibliography}{47}%
\makeatletter
\providecommand \@ifxundefined [1]{%
 \@ifx{#1\undefined}
}%
\providecommand \@ifnum [1]{%
 \ifnum #1\expandafter \@firstoftwo
 \else \expandafter \@secondoftwo
 \fi
}%
\providecommand \@ifx [1]{%
 \ifx #1\expandafter \@firstoftwo
 \else \expandafter \@secondoftwo
 \fi
}%
\providecommand \natexlab [1]{#1}%
\providecommand \enquote  [1]{``#1''}%
\providecommand \bibnamefont  [1]{#1}%
\providecommand \bibfnamefont [1]{#1}%
\providecommand \citenamefont [1]{#1}%
\providecommand \href@noop [0]{\@secondoftwo}%
\providecommand \href [0]{\begingroup \@sanitize@url \@href}%
\providecommand \@href[1]{\@@startlink{#1}\@@href}%
\providecommand \@@href[1]{\endgroup#1\@@endlink}%
\providecommand \@sanitize@url [0]{\catcode `\\12\catcode `\$12\catcode
  `\&12\catcode `\#12\catcode `\^12\catcode `\_12\catcode `\%12\relax}%
\providecommand \@@startlink[1]{}%
\providecommand \@@endlink[0]{}%
\providecommand \url  [0]{\begingroup\@sanitize@url \@url }%
\providecommand \@url [1]{\endgroup\@href {#1}{\urlprefix }}%
\providecommand \urlprefix  [0]{URL }%
\providecommand \Eprint [0]{\href }%
\providecommand \doibase [0]{http://dx.doi.org/}%
\providecommand \selectlanguage [0]{\@gobble}%
\providecommand \bibinfo  [0]{\@secondoftwo}%
\providecommand \bibfield  [0]{\@secondoftwo}%
\providecommand \translation [1]{[#1]}%
\providecommand \BibitemOpen [0]{}%
\providecommand \bibitemStop [0]{}%
\providecommand \bibitemNoStop [0]{.\EOS\space}%
\providecommand \EOS [0]{\spacefactor3000\relax}%
\providecommand \BibitemShut  [1]{\csname bibitem#1\endcsname}%
\let\auto@bib@innerbib\@empty
\bibitem [{\citenamefont {Bohr}(1949)}]{Bohr49}%
  \BibitemOpen
  \bibfield  {author} {\bibinfo {author} {\bibfnamefont {N.}~\bibnamefont
  {Bohr}},\ }\bibfield  {title} {\enquote {\bibinfo {title} {Discussion with
  einstein on epistemological problems in atomic physics},}\ }in\ \href
  {\doibase 10.1016/S1876-0503(08)70379-7} {\emph {\bibinfo {booktitle} {Albert
  Einstein: Philosopher--Scientist}}},\ \bibinfo {series} {The Library of
  Living Philosophers}, Vol.~\bibinfo {volume} {7},\ \bibinfo {editor} {edited
  by\ \bibinfo {editor} {\bibfnamefont {P.}~\bibnamefont {Schilpp}}}\ (\bibinfo
   {publisher} {Evanston},\ \bibinfo {year} {1949})\ pp.\ \bibinfo {pages}
  {201--241}\BibitemShut {NoStop}%
\bibitem [{\citenamefont {Landau}\ and\ \citenamefont
  {Peierls}(1931)}]{LandauPeierls31}%
  \BibitemOpen
  \bibfield  {author} {\bibinfo {author} {\bibfnamefont {L.}~\bibnamefont
  {Landau}}\ and\ \bibinfo {author} {\bibfnamefont {R.}~\bibnamefont
  {Peierls}},\ }\bibfield  {title} {\enquote {\bibinfo {title} {Erweiterung des
  {U}nbestimmtheitsprinzips für die relativistische {Q}uantentheorie},}\
  }\href {\doibase 10.1007/BF01391513} {\bibfield  {journal} {\bibinfo
  {journal} {Z. Physik}\ }\textbf {\bibinfo {volume} {69}},\ \bibinfo {pages}
  {56} (\bibinfo {year} {1931})}\BibitemShut {NoStop}%
\bibitem [{\citenamefont {Aharonov}\ and\ \citenamefont
  {Albert}(1980)}]{AharonovAlbert1980}%
  \BibitemOpen
  \bibfield  {author} {\bibinfo {author} {\bibfnamefont {Y.}~\bibnamefont
  {Aharonov}}\ and\ \bibinfo {author} {\bibfnamefont {D.~Z.}\ \bibnamefont
  {Albert}},\ }\bibfield  {title} {\enquote {\bibinfo {title} {States and
  observables in relativistic quantum field theories},}\ }\href {\doibase
  10.1103/PhysRevD.21.3316} {\bibfield  {journal} {\bibinfo  {journal} {Phys.
  Rev. D}\ }\textbf {\bibinfo {volume} {21}},\ \bibinfo {pages} {3316}
  (\bibinfo {year} {1980})}\BibitemShut {NoStop}%
\bibitem [{\citenamefont {Aharonov}\ and\ \citenamefont
  {Albert}(1981)}]{AharonovAlbert1981}%
  \BibitemOpen
  \bibfield  {author} {\bibinfo {author} {\bibfnamefont {Y.}~\bibnamefont
  {Aharonov}}\ and\ \bibinfo {author} {\bibfnamefont {D.~Z.}\ \bibnamefont
  {Albert}},\ }\bibfield  {title} {\enquote {\bibinfo {title} {Can we make
  sense out of the measurement process in relativistic quantum mechanics?}}\
  }\href {\doibase 10.1103/PhysRevD.24.359} {\bibfield  {journal} {\bibinfo
  {journal} {Phys. Rev. D}\ }\textbf {\bibinfo {volume} {24}},\ \bibinfo
  {pages} {359} (\bibinfo {year} {1981})}\BibitemShut {NoStop}%
\bibitem [{\citenamefont {Aharonov}\ and\ \citenamefont
  {Albert}(1984{\natexlab{a}})}]{AharonovAlbert1984I}%
  \BibitemOpen
  \bibfield  {author} {\bibinfo {author} {\bibfnamefont {Y.}~\bibnamefont
  {Aharonov}}\ and\ \bibinfo {author} {\bibfnamefont {D.~Z.}\ \bibnamefont
  {Albert}},\ }\bibfield  {title} {\enquote {\bibinfo {title} {Is the usual
  notion of time evolution adequate for quantum-mechanical systems? {I}},}\
  }\href {\doibase 10.1103/PhysRevD.29.223} {\bibfield  {journal} {\bibinfo
  {journal} {Phys. Rev. D}\ }\textbf {\bibinfo {volume} {29}},\ \bibinfo
  {pages} {223} (\bibinfo {year} {1984}{\natexlab{a}})}\BibitemShut {NoStop}%
\bibitem [{\citenamefont {Aharonov}\ and\ \citenamefont
  {Albert}(1984{\natexlab{b}})}]{AharonovAlbert1984II}%
  \BibitemOpen
  \bibfield  {author} {\bibinfo {author} {\bibfnamefont {Y.}~\bibnamefont
  {Aharonov}}\ and\ \bibinfo {author} {\bibfnamefont {D.~Z.}\ \bibnamefont
  {Albert}},\ }\bibfield  {title} {\enquote {\bibinfo {title} {Is the usual
  notion of time evolution adequate for quantum-mechanical systems? {II}.
  {R}elativistic considerations},}\ }\href {\doibase 10.1103/PhysRevD.29.228}
  {\bibfield  {journal} {\bibinfo  {journal} {Phys. Rev. D}\ }\textbf {\bibinfo
  {volume} {29}},\ \bibinfo {pages} {228} (\bibinfo {year}
  {1984}{\natexlab{b}})}\BibitemShut {NoStop}%
\bibitem [{\citenamefont {Aharonov}\ \emph {et~al.}(1986)\citenamefont
  {Aharonov}, \citenamefont {Albert},\ and\ \citenamefont {Vaidman}}]{AAV86}%
  \BibitemOpen
  \bibfield  {author} {\bibinfo {author} {\bibfnamefont {Y.}~\bibnamefont
  {Aharonov}}, \bibinfo {author} {\bibfnamefont {D.~Z.}\ \bibnamefont
  {Albert}}, \ and\ \bibinfo {author} {\bibfnamefont {L.}~\bibnamefont
  {Vaidman}},\ }\bibfield  {title} {\enquote {\bibinfo {title} {Measurement
  process in relativistic quantum theory},}\ }\href {\doibase
  10.1103/PhysRevD.34.1805} {\bibfield  {journal} {\bibinfo  {journal} {Phys.
  Rev. D}\ }\textbf {\bibinfo {volume} {34}},\ \bibinfo {pages} {1805}
  (\bibinfo {year} {1986})}\BibitemShut {NoStop}%
\bibitem [{\citenamefont {Popescu}\ and\ \citenamefont
  {Vaidman}(1994)}]{PopescuVaidman94}%
  \BibitemOpen
  \bibfield  {author} {\bibinfo {author} {\bibfnamefont {S.}~\bibnamefont
  {Popescu}}\ and\ \bibinfo {author} {\bibfnamefont {L.}~\bibnamefont
  {Vaidman}},\ }\bibfield  {title} {\enquote {\bibinfo {title} {Causality
  constraints on nonlocal quantum measurements},}\ }\href {\doibase
  10.1103/PhysRevA.49.4331} {\bibfield  {journal} {\bibinfo  {journal} {Phys.
  Rev. A}\ }\textbf {\bibinfo {volume} {49}},\ \bibinfo {pages} {4331}
  (\bibinfo {year} {1994})},\ \Eprint {http://arxiv.org/abs/hep-th/9306087}
  {arXiv:hep-th/9306087} \BibitemShut {NoStop}%
\bibitem [{\citenamefont {Groisman}\ and\ \citenamefont
  {Vaidman}(2001)}]{GroismanVaidman01}%
  \BibitemOpen
  \bibfield  {author} {\bibinfo {author} {\bibfnamefont {B.}~\bibnamefont
  {Groisman}}\ and\ \bibinfo {author} {\bibfnamefont {L.}~\bibnamefont
  {Vaidman}},\ }\bibfield  {title} {\enquote {\bibinfo {title} {Nonlocal
  variables with product-state eigenstates},}\ }\href {\doibase
  10.1088/0305-4470/34/35/313} {\bibfield  {journal} {\bibinfo  {journal}
  {Journal of Physics A: Mathematical and General}\ }\textbf {\bibinfo {volume}
  {34}},\ \bibinfo {pages} {6881} (\bibinfo {year} {2001})},\ \Eprint
  {http://arxiv.org/abs/quant-ph/0103084} {arXiv:quant-ph/0103084} \BibitemShut
  {NoStop}%
\bibitem [{\citenamefont {Vaidman}(2003)}]{Vaidman03}%
  \BibitemOpen
  \bibfield  {author} {\bibinfo {author} {\bibfnamefont {L.}~\bibnamefont
  {Vaidman}},\ }\bibfield  {title} {\enquote {\bibinfo {title} {Instantaneous
  measurement of nonlocal variables},}\ }\href {\doibase
  10.1103/PhysRevLett.90.010402} {\bibfield  {journal} {\bibinfo  {journal}
  {Phys. Rev. Lett.}\ }\textbf {\bibinfo {volume} {90}},\ \bibinfo {pages}
  {010402} (\bibinfo {year} {2003})},\ \Eprint
  {http://arxiv.org/abs/quant-ph/0111124} {arXiv:quant-ph/0111124} \BibitemShut
  {NoStop}%
\bibitem [{\citenamefont {Bennett}\ \emph {et~al.}(1993)\citenamefont
  {Bennett}, \citenamefont {Brassard}, \citenamefont {Cr{\'e}peau},
  \citenamefont {Jozsa}, \citenamefont {Peres},\ and\ \citenamefont
  {Wootters}}]{bennett1993teleporting}%
  \BibitemOpen
  \bibfield  {author} {\bibinfo {author} {\bibfnamefont {C.~H.}\ \bibnamefont
  {Bennett}}, \bibinfo {author} {\bibfnamefont {G.}~\bibnamefont {Brassard}},
  \bibinfo {author} {\bibfnamefont {C.}~\bibnamefont {Cr{\'e}peau}}, \bibinfo
  {author} {\bibfnamefont {R.}~\bibnamefont {Jozsa}}, \bibinfo {author}
  {\bibfnamefont {A.}~\bibnamefont {Peres}}, \ and\ \bibinfo {author}
  {\bibfnamefont {W.~K.}\ \bibnamefont {Wootters}},\ }\bibfield  {title}
  {\enquote {\bibinfo {title} {Teleporting an unknown quantum state via dual
  classical and {Einstein-Podolsky-Rosen} channels},}\ }\href {\doibase
  https://doi.org/10.1103/PhysRevLett.70.1895} {\bibfield  {journal} {\bibinfo
  {journal} {Physical review letters}\ }\textbf {\bibinfo {volume} {70}},\
  \bibinfo {pages} {1895} (\bibinfo {year} {1993})}\BibitemShut {NoStop}%
\bibitem [{\citenamefont {Chandran}\ \emph {et~al.}(2009)\citenamefont
  {Chandran}, \citenamefont {Goyal}, \citenamefont {Moriarty},\ and\
  \citenamefont {Ostrovsky}}]{CGMO09}%
  \BibitemOpen
  \bibfield  {author} {\bibinfo {author} {\bibfnamefont {N.}~\bibnamefont
  {Chandran}}, \bibinfo {author} {\bibfnamefont {V.}~\bibnamefont {Goyal}},
  \bibinfo {author} {\bibfnamefont {R.}~\bibnamefont {Moriarty}}, \ and\
  \bibinfo {author} {\bibfnamefont {R.}~\bibnamefont {Ostrovsky}},\ }\bibfield
  {title} {\enquote {\bibinfo {title} {Position based cryptography},}\ }in\
  \href {\doibase 10.1007/978-3-642-03356-8_23} {\emph {\bibinfo {booktitle}
  {Advances in Cryptology - CRYPTO 2009}}},\ \bibinfo {series} {Lecture Notes
  in Computer Science}, Vol.\ \bibinfo {volume} {5677},\ \bibinfo {editor}
  {edited by\ \bibinfo {editor} {\bibfnamefont {S.}~\bibnamefont {Halevi}}}\
  (\bibinfo  {publisher} {Springer Berlin Heidelberg},\ \bibinfo {year}
  {2009})\ pp.\ \bibinfo {pages} {391--407},\ \Eprint
  {http://eprint.iacr.org/2009/364} {IACR:2009/364} \BibitemShut {NoStop}%
\bibitem [{\citenamefont {Schaffner}(2019)}]{SchaffnerPBQCWeb}%
  \BibitemOpen
  \bibfield  {author} {\bibinfo {author} {\bibfnamefont {C.}~\bibnamefont
  {Schaffner}},\ }\href
  {http://homepages.cwi.nl/~schaffne/positionbasedqcrypto.php} {\enquote
  {\bibinfo {title} {Position based quantum cryptography},}\ }\bibinfo
  {howpublished} {Personal homepage of {C}hristian {S}chaffner} (\bibinfo
  {year} {2011--git 2019})\BibitemShut {NoStop}%
\bibitem [{\citenamefont {Kent}\ \emph {et~al.}(2006)\citenamefont {Kent},
  \citenamefont {Munro}, \citenamefont {Spiller},\ and\ \citenamefont
  {Beausoleil}}]{PatentKMSB06}%
  \BibitemOpen
  \bibfield  {author} {\bibinfo {author} {\bibfnamefont {A.~P.}\ \bibnamefont
  {Kent}}, \bibinfo {author} {\bibfnamefont {W.~J.}\ \bibnamefont {Munro}},
  \bibinfo {author} {\bibfnamefont {T.~P.}\ \bibnamefont {Spiller}}, \ and\
  \bibinfo {author} {\bibfnamefont {R.~G.}\ \bibnamefont {Beausoleil}},\ }\href
  {http://patft1.uspto.gov/netacgi/nph-Parser?patentnumber=7075438} {\enquote
  {\bibinfo {title} {Quantum tagging},}\ }\bibinfo {howpublished} {US patent
  7,075,438} (\bibinfo {year} {2006})\BibitemShut {NoStop}%
\bibitem [{\citenamefont {Kent}\ \emph {et~al.}(2011)\citenamefont {Kent},
  \citenamefont {Munro},\ and\ \citenamefont {Spiller}}]{KMS11}%
  \BibitemOpen
  \bibfield  {author} {\bibinfo {author} {\bibfnamefont {A.}~\bibnamefont
  {Kent}}, \bibinfo {author} {\bibfnamefont {W.~J.}\ \bibnamefont {Munro}}, \
  and\ \bibinfo {author} {\bibfnamefont {T.~P.}\ \bibnamefont {Spiller}},\
  }\bibfield  {title} {\enquote {\bibinfo {title} {Quantum tagging:
  Authenticating location via quantum information and relativistic signaling
  constraints},}\ }\href {\doibase 10.1103/PhysRevA.84.012326} {\bibfield
  {journal} {\bibinfo  {journal} {Phys. Rev. A}\ }\textbf {\bibinfo {volume}
  {84}},\ \bibinfo {pages} {012326} (\bibinfo {year} {2011})},\ \Eprint
  {http://arxiv.org/abs/1008.2147} {arXiv:1008.2147} \BibitemShut {NoStop}%
\bibitem [{\citenamefont {Chandran}\ \emph {et~al.}(2010)\citenamefont
  {Chandran}, \citenamefont {Fehr}, \citenamefont {Gelles}, \citenamefont
  {Goyal},\ and\ \citenamefont {Ostrovsky}}]{CFG+10}%
  \BibitemOpen
  \bibfield  {author} {\bibinfo {author} {\bibfnamefont {N.}~\bibnamefont
  {Chandran}}, \bibinfo {author} {\bibfnamefont {S.}~\bibnamefont {Fehr}},
  \bibinfo {author} {\bibfnamefont {R.}~\bibnamefont {Gelles}}, \bibinfo
  {author} {\bibfnamefont {V.}~\bibnamefont {Goyal}}, \ and\ \bibinfo {author}
  {\bibfnamefont {R.}~\bibnamefont {Ostrovsky}},\ }\href@noop {} {\enquote
  {\bibinfo {title} {Position-based quantum cryptography},}\ } (\bibinfo {year}
  {2010}),\ \bibinfo {note} {withdrawn and replaced
  by~\cite{BCF+14}}\BibitemShut {NoStop}%
\bibitem [{\citenamefont {Malaney}(2010{\natexlab{a}})}]{Malaney10a}%
  \BibitemOpen
  \bibfield  {author} {\bibinfo {author} {\bibfnamefont {R.~A.}\ \bibnamefont
  {Malaney}},\ }\bibfield  {title} {\enquote {\bibinfo {title}
  {Location-dependent communications using quantum entanglement},}\ }\href
  {\doibase 10.1103/PhysRevA.81.042319} {\bibfield  {journal} {\bibinfo
  {journal} {Phys. Rev. A}\ }\textbf {\bibinfo {volume} {81}},\ \bibinfo
  {pages} {042319} (\bibinfo {year} {2010}{\natexlab{a}})},\ \Eprint
  {http://arxiv.org/abs/1003.0949} {arXiv:1003.0949} \BibitemShut {NoStop}%
\bibitem [{\citenamefont {Malaney}(2010{\natexlab{b}})}]{Malaney10b}%
  \BibitemOpen
  \bibfield  {author} {\bibinfo {author} {\bibfnamefont {R.~A.}\ \bibnamefont
  {Malaney}},\ }\bibfield  {title} {\enquote {\bibinfo {title} {Quantum
  location verification in noisy channels},}\ }in\ \href {\doibase
  10.1109/GLOCOM.2010.5684009} {\emph {\bibinfo {booktitle} {Global
  Telecommunications Conference (GLOBECOM 2010), 2010 IEEE}}}\ (\bibinfo {year}
  {2010})\ pp.\ \bibinfo {pages} {1--6},\ \Eprint
  {http://arxiv.org/abs/1004.4689} {arXiv:1004.4689} \BibitemShut {NoStop}%
\bibitem [{\citenamefont {Buhrman}\ \emph {et~al.}(2014)\citenamefont
  {Buhrman}, \citenamefont {Chandran}, \citenamefont {Fehr}, \citenamefont
  {Gelles}, \citenamefont {Goyal}, \citenamefont {Ostrovsky},\ and\
  \citenamefont {Schaffner}}]{BCF+14}%
  \BibitemOpen
  \bibfield  {author} {\bibinfo {author} {\bibfnamefont {H.}~\bibnamefont
  {Buhrman}}, \bibinfo {author} {\bibfnamefont {N.}~\bibnamefont {Chandran}},
  \bibinfo {author} {\bibfnamefont {S.}~\bibnamefont {Fehr}}, \bibinfo {author}
  {\bibfnamefont {R.}~\bibnamefont {Gelles}}, \bibinfo {author} {\bibfnamefont
  {V.}~\bibnamefont {Goyal}}, \bibinfo {author} {\bibfnamefont
  {R.}~\bibnamefont {Ostrovsky}}, \ and\ \bibinfo {author} {\bibfnamefont
  {C.}~\bibnamefont {Schaffner}},\ }\bibfield  {title} {\enquote {\bibinfo
  {title} {Position-based quantum cryptography: Impossibility and
  constructions},}\ }\href {\doibase 10.1137/130913687} {\bibfield  {journal}
  {\bibinfo  {journal} {SIAM Journal on Computing}\ }\textbf {\bibinfo {volume}
  {43}},\ \bibinfo {pages} {150} (\bibinfo {year} {2014})},\ \Eprint
  {http://arxiv.org/abs/1009.2490} {arXiv:1009.2490} \BibitemShut {NoStop}%
\bibitem [{\citenamefont {Ishizaka}\ and\ \citenamefont
  {Hiroshima}(2009)}]{ishizaka2009quantum}%
  \BibitemOpen
  \bibfield  {author} {\bibinfo {author} {\bibfnamefont {S.}~\bibnamefont
  {Ishizaka}}\ and\ \bibinfo {author} {\bibfnamefont {T.}~\bibnamefont
  {Hiroshima}},\ }\bibfield  {title} {\enquote {\bibinfo {title} {Quantum
  teleportation scheme by selecting one of multiple output ports},}\ }\href
  {\doibase 10.1103/PhysRevA.79.042306} {\bibfield  {journal} {\bibinfo
  {journal} {Phys. Rev. A}\ }\textbf {\bibinfo {volume} {79}},\ \bibinfo
  {pages} {042306} (\bibinfo {year} {2009})},\ \Eprint
  {http://arxiv.org/abs/0901.2975} {arXiv:0901.2975} \BibitemShut {NoStop}%
\bibitem [{\citenamefont {Lau}\ and\ \citenamefont {Lo}(2011)}]{LauLo11}%
  \BibitemOpen
  \bibfield  {author} {\bibinfo {author} {\bibfnamefont {H.-K.}\ \bibnamefont
  {Lau}}\ and\ \bibinfo {author} {\bibfnamefont {H.-K.}\ \bibnamefont {Lo}},\
  }\bibfield  {title} {\enquote {\bibinfo {title} {Insecurity of position-based
  quantum-cryptography protocols against entanglement attacks},}\ }\href
  {\doibase 10.1103/PhysRevA.83.012322} {\bibfield  {journal} {\bibinfo
  {journal} {Phys. Rev. A}\ }\textbf {\bibinfo {volume} {83}},\ \bibinfo
  {pages} {012322} (\bibinfo {year} {2011})},\ \Eprint
  {http://arxiv.org/abs/1009.2256} {arXiv:1009.2256} \BibitemShut {NoStop}%
\bibitem [{\citenamefont {{Gonzales}}\ and\ \citenamefont
  {{Chitambar}}(2020)}]{GC20}%
  \BibitemOpen
  \bibfield  {author} {\bibinfo {author} {\bibfnamefont {A.}~\bibnamefont
  {{Gonzales}}}\ and\ \bibinfo {author} {\bibfnamefont {E.}~\bibnamefont
  {{Chitambar}}},\ }\bibfield  {title} {\enquote {\bibinfo {title} {Bounds on
  instantaneous nonlocal quantum computation},}\ }\href {\doibase
  10.1109/TIT.2019.2950190} {\bibfield  {journal} {\bibinfo  {journal} {IEEE
  Transactions on Information Theory}\ }\textbf {\bibinfo {volume} {66}},\
  \bibinfo {pages} {2951} (\bibinfo {year} {2020})},\ \Eprint
  {http://arxiv.org/abs/1810.00994} {arXiv:1810.00994} \BibitemShut {NoStop}%
\bibitem [{\citenamefont {Chakraborty}\ and\ \citenamefont
  {Leverrier}(2015)}]{CL15}%
  \BibitemOpen
  \bibfield  {author} {\bibinfo {author} {\bibfnamefont {K.}~\bibnamefont
  {Chakraborty}}\ and\ \bibinfo {author} {\bibfnamefont {A.}~\bibnamefont
  {Leverrier}},\ }\bibfield  {title} {\enquote {\bibinfo {title} {Practical
  position-based quantum cryptography},}\ }\href {\doibase
  10.1103/PhysRevA.92.052304} {\bibfield  {journal} {\bibinfo  {journal} {Phys.
  Rev. A}\ }\textbf {\bibinfo {volume} {92}},\ \bibinfo {pages} {052304}
  (\bibinfo {year} {2015})},\ \Eprint {http://arxiv.org/abs/1507.00626}
  {arXiv:1507.00626} \BibitemShut {NoStop}%
\bibitem [{\citenamefont {Speelman}(2016)}]{speelman2015instantaneous}%
  \BibitemOpen
  \bibfield  {author} {\bibinfo {author} {\bibfnamefont {F.}~\bibnamefont
  {Speelman}},\ }\bibfield  {title} {\enquote {\bibinfo {title} {Instantaneous
  non-local computation of low {T}-depth quantum circuits},}\ }in\ \href
  {\doibase 10.4230/LIPIcs.TQC.2016.9} {\emph {\bibinfo {booktitle} {11th
  Conference on the Theory of Quantum Computation, Communication and
  Cryptography (TQC 2016)}}},\ \bibinfo {series} {Leibniz International
  Proceedings in Informatics (LIPIcs)}, Vol.~\bibinfo {volume} {61},\ \bibinfo
  {editor} {edited by\ \bibinfo {editor} {\bibfnamefont {A.}~\bibnamefont
  {Broadbent}}}\ (\bibinfo  {publisher} {Schloss Dagstuhl--Leibniz-Zentrum fuer
  Informatik},\ \bibinfo {address} {Dagstuhl, Germany},\ \bibinfo {year}
  {2016})\ pp.\ \bibinfo {pages} {9:1--9:24}\BibitemShut {NoStop}%
\bibitem [{\citenamefont {Buhrman}\ \emph {et~al.}(2013)\citenamefont
  {Buhrman}, \citenamefont {Fehr}, \citenamefont {Schaffner},\ and\
  \citenamefont {Speelman}}]{BFSS13}%
  \BibitemOpen
  \bibfield  {author} {\bibinfo {author} {\bibfnamefont {H.}~\bibnamefont
  {Buhrman}}, \bibinfo {author} {\bibfnamefont {S.}~\bibnamefont {Fehr}},
  \bibinfo {author} {\bibfnamefont {C.}~\bibnamefont {Schaffner}}, \ and\
  \bibinfo {author} {\bibfnamefont {F.}~\bibnamefont {Speelman}},\ }\bibfield
  {title} {\enquote {\bibinfo {title} {The garden-hose model},}\ }in\ \href
  {\doibase 10.1145/2422436.2422455} {\emph {\bibinfo {booktitle} {Proceedings
  of the 4th Conference on Innovations in Theoretical Computer Science}}},\
  \bibinfo {series and number} {ITCS ’13}\ (\bibinfo  {publisher}
  {Association for Computing Machinery},\ \bibinfo {address} {New York, NY,
  USA},\ \bibinfo {year} {2013})\ p.\ \bibinfo {pages} {145–158}\BibitemShut
  {NoStop}%
\bibitem [{\citenamefont {Klauck}\ and\ \citenamefont
  {Podder}(2014)}]{KlauckPodder14}%
  \BibitemOpen
  \bibfield  {author} {\bibinfo {author} {\bibfnamefont {H.}~\bibnamefont
  {Klauck}}\ and\ \bibinfo {author} {\bibfnamefont {S.}~\bibnamefont
  {Podder}},\ }\bibfield  {title} {\enquote {\bibinfo {title} {{New Bounds for
  the Garden-Hose Model}},}\ }in\ \href {\doibase
  10.4230/LIPIcs.FSTTCS.2014.481} {\emph {\bibinfo {booktitle} {34th
  International Conference on Foundation of Software Technology and Theoretical
  Computer Science (FSTTCS 2014)}}},\ \bibinfo {series} {Leibniz International
  Proceedings in Informatics (LIPIcs)}, Vol.~\bibinfo {volume} {29},\ \bibinfo
  {editor} {edited by\ \bibinfo {editor} {\bibfnamefont {V.}~\bibnamefont
  {Raman}}\ and\ \bibinfo {editor} {\bibfnamefont {S.~P.}\ \bibnamefont
  {Suresh}}}\ (\bibinfo  {publisher} {Schloss Dagstuhl--Leibniz-Zentrum fuer
  Informatik},\ \bibinfo {address} {Dagstuhl, Germany},\ \bibinfo {year}
  {2014})\ pp.\ \bibinfo {pages} {481--492}\BibitemShut {NoStop}%
\bibitem [{\citenamefont {Unruh}(2014)}]{Unruh14}%
  \BibitemOpen
  \bibfield  {author} {\bibinfo {author} {\bibfnamefont {D.}~\bibnamefont
  {Unruh}},\ }\bibfield  {title} {\enquote {\bibinfo {title} {Quantum position
  verification in the random oracle model},}\ }in\ \href {\doibase
  10.1007/978-3-662-44381-1_1} {\emph {\bibinfo {booktitle} {Advances in
  Cryptology – CRYPTO 2014}}},\ \bibinfo {series} {Lecture Notes in Computer
  Science}, Vol.\ \bibinfo {volume} {8617},\ \bibinfo {editor} {edited by\
  \bibinfo {editor} {\bibfnamefont {J.~A.}\ \bibnamefont {Garay}}\ and\
  \bibinfo {editor} {\bibfnamefont {R.}~\bibnamefont {Gennaro}}}\ (\bibinfo
  {publisher} {Springer Berlin Heidelberg},\ \bibinfo {year} {2014})\ pp.\
  \bibinfo {pages} {1--18},\ \Eprint {http://eprint.iacr.org/2014/118}
  {IACR:2014/118} \BibitemShut {NoStop}%
\bibitem [{\citenamefont {Tomamichel}\ \emph {et~al.}(2013)\citenamefont
  {Tomamichel}, \citenamefont {Fehr}, \citenamefont {Kaniewski},\ and\
  \citenamefont {Wehner}}]{TFKW13}%
  \BibitemOpen
  \bibfield  {author} {\bibinfo {author} {\bibfnamefont {M.}~\bibnamefont
  {Tomamichel}}, \bibinfo {author} {\bibfnamefont {S.}~\bibnamefont {Fehr}},
  \bibinfo {author} {\bibfnamefont {J.}~\bibnamefont {Kaniewski}}, \ and\
  \bibinfo {author} {\bibfnamefont {S.}~\bibnamefont {Wehner}},\ }\bibfield
  {title} {\enquote {\bibinfo {title} {A monogamy-of-entanglement game with
  applications to device-independent quantum cryptography},}\ }\href {\doibase
  10.1088/1367-2630/15/10/103002} {\bibfield  {journal} {\bibinfo  {journal}
  {New Journal of Physics}\ }\textbf {\bibinfo {volume} {15}},\ \bibinfo
  {pages} {103002} (\bibinfo {year} {2013})},\ \Eprint
  {http://arxiv.org/abs/1210.4359} {arXiv:1210.4359} \BibitemShut {NoStop}%
\bibitem [{\citenamefont {Ribeiro}\ and\ \citenamefont
  {Grosshans}(2015)}]{RG15}%
  \BibitemOpen
  \bibfield  {author} {\bibinfo {author} {\bibfnamefont {J.}~\bibnamefont
  {Ribeiro}}\ and\ \bibinfo {author} {\bibfnamefont {F.}~\bibnamefont
  {Grosshans}},\ }\href@noop {} {\enquote {\bibinfo {title} {A tight lower
  bound for the {BB84}-states quantum-position-verification protocol},}\ }
  (\bibinfo {year} {2015}),\ \Eprint {http://arxiv.org/abs/1504.07171}
  {arXiv:1504.07171} \BibitemShut {NoStop}%
\bibitem [{\citenamefont {Bennett}\ and\ \citenamefont
  {Brassard}(1984)}]{BB84}%
  \BibitemOpen
  \bibfield  {author} {\bibinfo {author} {\bibfnamefont {C.~H.}\ \bibnamefont
  {Bennett}}\ and\ \bibinfo {author} {\bibfnamefont {G.}~\bibnamefont
  {Brassard}},\ }\bibfield  {title} {\enquote {\bibinfo {title} {Quantum
  cryptography: Public key distribution and coin tossing},}\ }in\ \href
  {http://researcher.watson.ibm.com/researcher/files/us-bennetc/BB84highest.pdf}
  {\emph {\bibinfo {booktitle} {IEEE International Conference on Computers,
  Systems and Signal Processing}}},\ Vol.\ \bibinfo {volume} {175}\ (\bibinfo
  {address} {Bangalore, India},\ \bibinfo {year} {1984})\ p.~\bibinfo {pages}
  {8}\BibitemShut {NoStop}%
\bibitem [{\citenamefont {Pugh}\ \emph {et~al.}(2017)\citenamefont {Pugh},
  \citenamefont {Kaiser}, \citenamefont {Bourgoin}, \citenamefont {Jin},
  \citenamefont {Sultana}, \citenamefont {Agne}, \citenamefont {Anisimova},
  \citenamefont {Makarov}, \citenamefont {Choi}, \citenamefont {Higgins},\ and\
  \citenamefont {Jennewein}}]{Pugh+17}%
  \BibitemOpen
  \bibfield  {author} {\bibinfo {author} {\bibfnamefont {C.~J.}\ \bibnamefont
  {Pugh}}, \bibinfo {author} {\bibfnamefont {S.}~\bibnamefont {Kaiser}},
  \bibinfo {author} {\bibfnamefont {J.-P.}\ \bibnamefont {Bourgoin}}, \bibinfo
  {author} {\bibfnamefont {J.}~\bibnamefont {Jin}}, \bibinfo {author}
  {\bibfnamefont {N.}~\bibnamefont {Sultana}}, \bibinfo {author} {\bibfnamefont
  {S.}~\bibnamefont {Agne}}, \bibinfo {author} {\bibfnamefont {E.}~\bibnamefont
  {Anisimova}}, \bibinfo {author} {\bibfnamefont {V.}~\bibnamefont {Makarov}},
  \bibinfo {author} {\bibfnamefont {E.}~\bibnamefont {Choi}}, \bibinfo {author}
  {\bibfnamefont {B.~L.}\ \bibnamefont {Higgins}}, \ and\ \bibinfo {author}
  {\bibfnamefont {T.}~\bibnamefont {Jennewein}},\ }\bibfield  {title} {\enquote
  {\bibinfo {title} {Airborne demonstration of a quantum key distribution
  receiver payload},}\ }\href {\doibase 10.1088/2058-9565/aa701f} {\bibfield
  {journal} {\bibinfo  {journal} {Quantum Science and Technology}\ }\textbf
  {\bibinfo {volume} {2}},\ \bibinfo {pages} {024009} (\bibinfo {year}
  {2017})},\ \Eprint {http://arxiv.org/abs/1612.06396} {arXiv:1612.06396}
  \BibitemShut {NoStop}%
\bibitem [{\citenamefont {Liao}\ \emph {et~al.}(2017)\citenamefont {Liao},
  \citenamefont {Cai}, \citenamefont {Liu}, \citenamefont {Zhang},
  \citenamefont {Li}, \citenamefont {Ren}, \citenamefont {Yin}, \citenamefont
  {Shen}, \citenamefont {Cao}, \citenamefont {Li}, \citenamefont {Li},
  \citenamefont {Chen}, \citenamefont {Sun}, \citenamefont {Jia}, \citenamefont
  {Wu}, \citenamefont {Jiang}, \citenamefont {Wang}, \citenamefont {Huang},
  \citenamefont {Wang}, \citenamefont {Zhou}, \citenamefont {Deng},
  \citenamefont {Xi}, \citenamefont {Ma}, \citenamefont {Hu}, \citenamefont
  {Zhang}, \citenamefont {Chen}, \citenamefont {Liu}, \citenamefont {Wang},
  \citenamefont {Zhu}, \citenamefont {Lu}, , \citenamefont {Shu}, \citenamefont
  {Peng}, \citenamefont {Wang},\ and\ \citenamefont {Pan}}]{Liao+17}%
  \BibitemOpen
  \bibfield  {author} {\bibinfo {author} {\bibfnamefont {S.-K.}\ \bibnamefont
  {Liao}}, \bibinfo {author} {\bibfnamefont {W.-Q.}\ \bibnamefont {Cai}},
  \bibinfo {author} {\bibfnamefont {W.-Y.}\ \bibnamefont {Liu}}, \bibinfo
  {author} {\bibfnamefont {L.}~\bibnamefont {Zhang}}, \bibinfo {author}
  {\bibfnamefont {Y.}~\bibnamefont {Li}}, \bibinfo {author} {\bibfnamefont
  {J.-G.}\ \bibnamefont {Ren}}, \bibinfo {author} {\bibfnamefont
  {J.}~\bibnamefont {Yin}}, \bibinfo {author} {\bibfnamefont {Q.}~\bibnamefont
  {Shen}}, \bibinfo {author} {\bibfnamefont {Y.}~\bibnamefont {Cao}}, \bibinfo
  {author} {\bibfnamefont {Z.-P.}\ \bibnamefont {Li}}, \bibinfo {author}
  {\bibfnamefont {F.-Z.}\ \bibnamefont {Li}}, \bibinfo {author} {\bibfnamefont
  {X.-W.}\ \bibnamefont {Chen}}, \bibinfo {author} {\bibfnamefont {L.-H.}\
  \bibnamefont {Sun}}, \bibinfo {author} {\bibfnamefont {J.-J.}\ \bibnamefont
  {Jia}}, \bibinfo {author} {\bibfnamefont {J.-C.}\ \bibnamefont {Wu}},
  \bibinfo {author} {\bibfnamefont {X.-J.}\ \bibnamefont {Jiang}}, \bibinfo
  {author} {\bibfnamefont {J.-F.}\ \bibnamefont {Wang}}, \bibinfo {author}
  {\bibfnamefont {Y.-M.}\ \bibnamefont {Huang}}, \bibinfo {author}
  {\bibfnamefont {Q.}~\bibnamefont {Wang}}, \bibinfo {author} {\bibfnamefont
  {Y.-L.}\ \bibnamefont {Zhou}}, \bibinfo {author} {\bibfnamefont
  {L.}~\bibnamefont {Deng}}, \bibinfo {author} {\bibfnamefont {T.}~\bibnamefont
  {Xi}}, \bibinfo {author} {\bibfnamefont {L.}~\bibnamefont {Ma}}, \bibinfo
  {author} {\bibfnamefont {T.}~\bibnamefont {Hu}}, \bibinfo {author}
  {\bibfnamefont {Q.}~\bibnamefont {Zhang}}, \bibinfo {author} {\bibfnamefont
  {Y.-A.}\ \bibnamefont {Chen}}, \bibinfo {author} {\bibfnamefont {N.-L.}\
  \bibnamefont {Liu}}, \bibinfo {author} {\bibfnamefont {X.-B.}\ \bibnamefont
  {Wang}}, \bibinfo {author} {\bibfnamefont {Z.-C.}\ \bibnamefont {Zhu}},
  \bibinfo {author} {\bibfnamefont {C.-Y.}\ \bibnamefont {Lu}}, , \bibinfo
  {author} {\bibfnamefont {R.}~\bibnamefont {Shu}}, \bibinfo {author}
  {\bibfnamefont {C.-Z.}\ \bibnamefont {Peng}}, \bibinfo {author}
  {\bibfnamefont {J.-Y.}\ \bibnamefont {Wang}}, \ and\ \bibinfo {author}
  {\bibfnamefont {J.-W.}\ \bibnamefont {Pan}},\ }\bibfield  {title} {\enquote
  {\bibinfo {title} {Satellite-to-ground quantum key distribution},}\ }\href
  {\doibase 10.1038/nature23655} {\bibfield  {journal} {\bibinfo  {journal}
  {Nature}\ }\textbf {\bibinfo {volume} {549}},\ \bibinfo {pages} {43}
  (\bibinfo {year} {2017})},\ \Eprint {http://arxiv.org/abs/1707.00542}
  {arXiv:1707.00542} \BibitemShut {NoStop}%
\bibitem [{\citenamefont {Avesani}\ \emph {et~al.}(2019)\citenamefont
  {Avesani}, \citenamefont {Calderaro}, \citenamefont {Schiavon}, \citenamefont
  {Stanco}, \citenamefont {Agnesi}, \citenamefont {Santamato}, \citenamefont
  {Zahidy}, \citenamefont {Scriminich}, \citenamefont {Foletto}, \citenamefont
  {Contestabile}, \citenamefont {Chiesa}, \citenamefont {Rotta}, \citenamefont
  {Artiglia}, \citenamefont {Montanaro}, \citenamefont {Romagnoli},
  \citenamefont {Sorianello}, \citenamefont {Vedovato}, \citenamefont
  {Vallone},\ and\ \citenamefont {Villoresi}}]{Avesani+19}%
  \BibitemOpen
  \bibfield  {author} {\bibinfo {author} {\bibfnamefont {M.}~\bibnamefont
  {Avesani}}, \bibinfo {author} {\bibfnamefont {L.}~\bibnamefont {Calderaro}},
  \bibinfo {author} {\bibfnamefont {M.}~\bibnamefont {Schiavon}}, \bibinfo
  {author} {\bibfnamefont {A.}~\bibnamefont {Stanco}}, \bibinfo {author}
  {\bibfnamefont {C.}~\bibnamefont {Agnesi}}, \bibinfo {author} {\bibfnamefont
  {A.}~\bibnamefont {Santamato}}, \bibinfo {author} {\bibfnamefont
  {M.}~\bibnamefont {Zahidy}}, \bibinfo {author} {\bibfnamefont
  {A.}~\bibnamefont {Scriminich}}, \bibinfo {author} {\bibfnamefont
  {G.}~\bibnamefont {Foletto}}, \bibinfo {author} {\bibfnamefont
  {G.}~\bibnamefont {Contestabile}}, \bibinfo {author} {\bibfnamefont
  {M.}~\bibnamefont {Chiesa}}, \bibinfo {author} {\bibfnamefont
  {D.}~\bibnamefont {Rotta}}, \bibinfo {author} {\bibfnamefont
  {M.}~\bibnamefont {Artiglia}}, \bibinfo {author} {\bibfnamefont
  {A.}~\bibnamefont {Montanaro}}, \bibinfo {author} {\bibfnamefont
  {M.}~\bibnamefont {Romagnoli}}, \bibinfo {author} {\bibfnamefont
  {V.}~\bibnamefont {Sorianello}}, \bibinfo {author} {\bibfnamefont
  {F.}~\bibnamefont {Vedovato}}, \bibinfo {author} {\bibfnamefont
  {G.}~\bibnamefont {Vallone}}, \ and\ \bibinfo {author} {\bibfnamefont
  {P.}~\bibnamefont {Villoresi}},\ }\href@noop {} {\enquote {\bibinfo {title}
  {Full daylight quantum-key-distribution at 1550 nm enabled by integrated
  silicon photonics},}\ } (\bibinfo {year} {2019}),\ \Eprint
  {http://arxiv.org/abs/1907.10039} {arXiv:1907.10039} \BibitemShut {NoStop}%
\bibitem [{\citenamefont {Beigi}\ and\ \citenamefont
  {König}(2011)}]{BeigiKoenig11}%
  \BibitemOpen
  \bibfield  {author} {\bibinfo {author} {\bibfnamefont {S.}~\bibnamefont
  {Beigi}}\ and\ \bibinfo {author} {\bibfnamefont {R.}~\bibnamefont {König}},\
  }\bibfield  {title} {\enquote {\bibinfo {title} {Simplified instantaneous
  non-local quantum computation with applications to position-based
  cryptography},}\ }\href {\doibase 10.1088/1367-2630/13/9/093036} {\bibfield
  {journal} {\bibinfo  {journal} {New Journal of Physics}\ }\textbf {\bibinfo
  {volume} {13}},\ \bibinfo {pages} {093036} (\bibinfo {year} {2011})},\
  \Eprint {http://arxiv.org/abs/1101.1065} {arXiv:1101.1065} \BibitemShut
  {NoStop}%
\bibitem [{\citenamefont {Wilde}(2013)}]{wilde_2013}%
  \BibitemOpen
  \bibfield  {author} {\bibinfo {author} {\bibfnamefont {M.~M.}\ \bibnamefont
  {Wilde}},\ }\href {\doibase 10.1017/CBO9781139525343} {\emph {\bibinfo
  {title} {Quantum Information Theory}}}\ (\bibinfo  {publisher} {Cambridge
  University Press},\ \bibinfo {year} {2013})\ \Eprint
  {http://arxiv.org/abs/1106.1445} {arXiv:1106.1445} \BibitemShut {NoStop}%
\bibitem [{\citenamefont {{Virtanen}}\ \emph {et~al.}(2020)\citenamefont
  {{Virtanen}}, \citenamefont {{Gommers}}, \citenamefont {{Oliphant}},
  \citenamefont {{Haberland}}, \citenamefont {{Reddy}}, \citenamefont
  {{Cournapeau}}, \citenamefont {{Burovski}}, \citenamefont {{Peterson}},
  \citenamefont {{Weckesser}}, \citenamefont {{Bright}}, \citenamefont {{van
  der Walt}}, \citenamefont {{Brett}}, \citenamefont {{Wilson}}, \citenamefont
  {{Jarrod Millman}}, \citenamefont {{Mayorov}}, \citenamefont {{Nelson}},
  \citenamefont {{Jones}}, \citenamefont {{Kern}}, \citenamefont {{Larson}},
  \citenamefont {{Carey}}, \citenamefont {{Polat}}, \citenamefont {{Feng}},
  \citenamefont {{Moore}}, \citenamefont {{VanderPlas}}, \citenamefont
  {{Laxalde}}, \citenamefont {{Perktold}}, \citenamefont {{Cimrman}},
  \citenamefont {{Henriksen}}, \citenamefont {{Quintero}}, \citenamefont
  {{Harris}}, \citenamefont {{Archibald}}, \citenamefont {{Ribeiro}},
  \citenamefont {{Pedregosa}}, \citenamefont {{van Mulbregt}},\ and\
  \citenamefont {{SciPy 1.0 Contributors}}}]{2020SciPy-NMeth}%
  \BibitemOpen
  \bibfield  {author} {\bibinfo {author} {\bibfnamefont {P.}~\bibnamefont
  {{Virtanen}}}, \bibinfo {author} {\bibfnamefont {R.}~\bibnamefont
  {{Gommers}}}, \bibinfo {author} {\bibfnamefont {T.~E.}\ \bibnamefont
  {{Oliphant}}}, \bibinfo {author} {\bibfnamefont {M.}~\bibnamefont
  {{Haberland}}}, \bibinfo {author} {\bibfnamefont {T.}~\bibnamefont
  {{Reddy}}}, \bibinfo {author} {\bibfnamefont {D.}~\bibnamefont
  {{Cournapeau}}}, \bibinfo {author} {\bibfnamefont {E.}~\bibnamefont
  {{Burovski}}}, \bibinfo {author} {\bibfnamefont {P.}~\bibnamefont
  {{Peterson}}}, \bibinfo {author} {\bibfnamefont {W.}~\bibnamefont
  {{Weckesser}}}, \bibinfo {author} {\bibfnamefont {J.}~\bibnamefont
  {{Bright}}}, \bibinfo {author} {\bibfnamefont {S.~J.}\ \bibnamefont {{van der
  Walt}}}, \bibinfo {author} {\bibfnamefont {M.}~\bibnamefont {{Brett}}},
  \bibinfo {author} {\bibfnamefont {J.}~\bibnamefont {{Wilson}}}, \bibinfo
  {author} {\bibfnamefont {K.}~\bibnamefont {{Jarrod Millman}}}, \bibinfo
  {author} {\bibfnamefont {N.}~\bibnamefont {{Mayorov}}}, \bibinfo {author}
  {\bibfnamefont {A.~R.~J.}\ \bibnamefont {{Nelson}}}, \bibinfo {author}
  {\bibfnamefont {E.}~\bibnamefont {{Jones}}}, \bibinfo {author} {\bibfnamefont
  {R.}~\bibnamefont {{Kern}}}, \bibinfo {author} {\bibfnamefont
  {E.}~\bibnamefont {{Larson}}}, \bibinfo {author} {\bibfnamefont
  {C.}~\bibnamefont {{Carey}}}, \bibinfo {author} {\bibfnamefont
  {{\.I}.}~\bibnamefont {{Polat}}}, \bibinfo {author} {\bibfnamefont
  {Y.}~\bibnamefont {{Feng}}}, \bibinfo {author} {\bibfnamefont {E.~W.}\
  \bibnamefont {{Moore}}}, \bibinfo {author} {\bibfnamefont {J.}~\bibnamefont
  {{VanderPlas}}}, \bibinfo {author} {\bibfnamefont {D.}~\bibnamefont
  {{Laxalde}}}, \bibinfo {author} {\bibfnamefont {J.}~\bibnamefont
  {{Perktold}}}, \bibinfo {author} {\bibfnamefont {R.}~\bibnamefont
  {{Cimrman}}}, \bibinfo {author} {\bibfnamefont {I.}~\bibnamefont
  {{Henriksen}}}, \bibinfo {author} {\bibfnamefont {E.~A.}\ \bibnamefont
  {{Quintero}}}, \bibinfo {author} {\bibfnamefont {C.~R.}\ \bibnamefont
  {{Harris}}}, \bibinfo {author} {\bibfnamefont {A.~M.}\ \bibnamefont
  {{Archibald}}}, \bibinfo {author} {\bibfnamefont {A.~H.}\ \bibnamefont
  {{Ribeiro}}}, \bibinfo {author} {\bibfnamefont {F.}~\bibnamefont
  {{Pedregosa}}}, \bibinfo {author} {\bibfnamefont {P.}~\bibnamefont {{van
  Mulbregt}}}, \ and\ \bibinfo {author} {\bibnamefont {{SciPy 1.0
  Contributors}}},\ }\bibfield  {title} {\enquote {\bibinfo {title} {{SciPy
  1.0}: Fundamental algorithms for scientific computing in {Python}},}\ }\href
  {\doibase 10.1038/s41592-019-0686-2} {\bibfield  {journal} {\bibinfo
  {journal} {Nature Methods}\ }\textbf {\bibinfo {volume} {17}},\ \bibinfo
  {pages} {261} (\bibinfo {year} {2020})}\BibitemShut {NoStop}%
\bibitem [{\citenamefont {Schaffner}(2014)}]{schaffpriv}%
  \BibitemOpen
  \bibfield  {author} {\bibinfo {author} {\bibfnamefont {C.}~\bibnamefont
  {Schaffner}},\ }\href@noop {} {}\bibinfo {howpublished} {private
  communication} (\bibinfo {year} {2014})\BibitemShut {NoStop}%
\bibitem [{\citenamefont {Kay}(2018)}]{kay2018tutorial}%
  \BibitemOpen
  \bibfield  {author} {\bibinfo {author} {\bibfnamefont {A.}~\bibnamefont
  {Kay}},\ }\href@noop {} {\enquote {\bibinfo {title} {Tutorial on the
  {Quantikz} package},}\ } (\bibinfo {year} {2018}),\ \Eprint
  {http://arxiv.org/abs/1809.03842} {arXiv:1809.03842} \BibitemShut {NoStop}%
\bibitem [{\citenamefont {Raussendorf}\ \emph {et~al.}(2007)\citenamefont
  {Raussendorf}, \citenamefont {Harrington},\ and\ \citenamefont
  {Goyal}}]{Raussendorf_2007}%
  \BibitemOpen
  \bibfield  {author} {\bibinfo {author} {\bibfnamefont {R.}~\bibnamefont
  {Raussendorf}}, \bibinfo {author} {\bibfnamefont {J.}~\bibnamefont
  {Harrington}}, \ and\ \bibinfo {author} {\bibfnamefont {K.}~\bibnamefont
  {Goyal}},\ }\bibfield  {title} {\enquote {\bibinfo {title} {Topological
  fault-tolerance in cluster state quantum computation},}\ }\href {\doibase
  10.1088/1367-2630/9/6/199} {\bibfield  {journal} {\bibinfo  {journal} {New
  Journal of Physics}\ }\textbf {\bibinfo {volume} {9}},\ \bibinfo {pages}
  {199} (\bibinfo {year} {2007})}\BibitemShut {NoStop}%
\bibitem [{\citenamefont {Kraus}\ and\ \citenamefont
  {Cirac}(2001)}]{Kraus2001}%
  \BibitemOpen
  \bibfield  {author} {\bibinfo {author} {\bibfnamefont {B.}~\bibnamefont
  {Kraus}}\ and\ \bibinfo {author} {\bibfnamefont {J.~I.}\ \bibnamefont
  {Cirac}},\ }\bibfield  {title} {\enquote {\bibinfo {title} {Optimal creation
  of entanglement using a two-qubit gate},}\ }\href {\doibase
  10.1103/PhysRevA.63.062309} {\bibfield  {journal} {\bibinfo  {journal} {Phys.
  Rev. A}\ }\textbf {\bibinfo {volume} {63}},\ \bibinfo {pages} {062309}
  (\bibinfo {year} {2001})},\ \Eprint {http://arxiv.org/abs/0011050}
  {arXiv:0011050} \BibitemShut {NoStop}%
\bibitem [{\citenamefont {Drury}\ and\ \citenamefont
  {Love}(2008)}]{Drury_2008}%
  \BibitemOpen
  \bibfield  {author} {\bibinfo {author} {\bibfnamefont {B.}~\bibnamefont
  {Drury}}\ and\ \bibinfo {author} {\bibfnamefont {P.}~\bibnamefont {Love}},\
  }\bibfield  {title} {\enquote {\bibinfo {title} {Constructive quantum
  {Shannon} decomposition from {Cartan} involutions},}\ }\href {\doibase
  10.1088/1751-8113/41/39/395305} {\bibfield  {journal} {\bibinfo  {journal}
  {Journal of Physics A: Mathematical and Theoretical}\ }\textbf {\bibinfo
  {volume} {41}},\ \bibinfo {pages} {395305} (\bibinfo {year} {2008})},\
  \Eprint {http://arxiv.org/abs/0806.4015} {arXiv:0806.4015} \BibitemShut
  {NoStop}%
\bibitem [{\citenamefont {Parrilo}(2003)}]{parrilo2003semidefinite}%
  \BibitemOpen
  \bibfield  {author} {\bibinfo {author} {\bibfnamefont {P.~A.}\ \bibnamefont
  {Parrilo}},\ }\bibfield  {title} {\enquote {\bibinfo {title} {Semidefinite
  programming relaxations for semialgebraic problems},}\ }\href {\doibase
  10.1007/s10107-003-0387-5} {\bibfield  {journal} {\bibinfo  {journal}
  {Mathematical programming}\ }\textbf {\bibinfo {volume} {96}},\ \bibinfo
  {pages} {293} (\bibinfo {year} {2003})}\BibitemShut {NoStop}%
\bibitem [{\citenamefont {Pardalos}\ and\ \citenamefont
  {Schnitger}(1988)}]{PARDALOS198833}%
  \BibitemOpen
  \bibfield  {author} {\bibinfo {author} {\bibfnamefont {P.}~\bibnamefont
  {Pardalos}}\ and\ \bibinfo {author} {\bibfnamefont {G.}~\bibnamefont
  {Schnitger}},\ }\bibfield  {title} {\enquote {\bibinfo {title} {Checking
  local optimality in constrained quadratic programming is {NP-hard}},}\ }\href
  {\doibase 10.1016/0167-6377(88)90049-1} {\bibfield  {journal} {\bibinfo
  {journal} {Operations Research Letters}\ }\textbf {\bibinfo {volume} {7}},\
  \bibinfo {pages} {33 } (\bibinfo {year} {1988})}\BibitemShut {NoStop}%
\bibitem [{\citenamefont {W{\"a}chter}\ and\ \citenamefont
  {Biegler}(2006)}]{wachter2006implementation}%
  \BibitemOpen
  \bibfield  {author} {\bibinfo {author} {\bibfnamefont {A.}~\bibnamefont
  {W{\"a}chter}}\ and\ \bibinfo {author} {\bibfnamefont {L.~T.}\ \bibnamefont
  {Biegler}},\ }\bibfield  {title} {\enquote {\bibinfo {title} {On the
  implementation of an interior-point filter line-search algorithm for
  large-scale nonlinear programming},}\ }\href {\doibase
  10.1007/s10107-004-0559-y} {\bibfield  {journal} {\bibinfo  {journal}
  {Mathematical programming}\ }\textbf {\bibinfo {volume} {106}},\ \bibinfo
  {pages} {25} (\bibinfo {year} {2006})}\BibitemShut {NoStop}%
\bibitem [{\citenamefont {Byrd}\ \emph {et~al.}(1995)\citenamefont {Byrd},
  \citenamefont {Lu}, \citenamefont {Nocedal},\ and\ \citenamefont
  {Zhu}}]{BLNZ95}%
  \BibitemOpen
  \bibfield  {author} {\bibinfo {author} {\bibfnamefont {R.~H.}\ \bibnamefont
  {Byrd}}, \bibinfo {author} {\bibfnamefont {P.}~\bibnamefont {Lu}}, \bibinfo
  {author} {\bibfnamefont {J.}~\bibnamefont {Nocedal}}, \ and\ \bibinfo
  {author} {\bibfnamefont {C.}~\bibnamefont {Zhu}},\ }\bibfield  {title}
  {\enquote {\bibinfo {title} {A limited memory algorithm for bound constrained
  optimization},}\ }\href {\doibase 10.1137/0916069} {\bibfield  {journal}
  {\bibinfo  {journal} {SIAM Journal on Scientific Computing}\ }\textbf
  {\bibinfo {volume} {16}},\ \bibinfo {pages} {1190} (\bibinfo {year}
  {1995})}\BibitemShut {NoStop}%
\bibitem [{\citenamefont {Cayley}(1846)}]{Cayley1846}%
  \BibitemOpen
  \bibfield  {author} {\bibinfo {author} {\bibfnamefont {A.}~\bibnamefont
  {Cayley}},\ }\bibfield  {title} {\enquote {\bibinfo {title} {Sur quelques
  propri\'{e}t\'{e}s des d\'{e}terminants gauches},}\ }\href
  {http://eudml.org/doc/147336} {\bibfield  {journal} {\bibinfo  {journal}
  {Journal f\"{u}r die reine und angewandte Mathematik}\ }\textbf {\bibinfo
  {volume} {32}},\ \bibinfo {pages} {119} (\bibinfo {year} {1846})}\BibitemShut
  {NoStop}%
\bibitem [{\citenamefont {Zhu}(2017)}]{zhu2017riemannian}%
  \BibitemOpen
  \bibfield  {author} {\bibinfo {author} {\bibfnamefont {X.}~\bibnamefont
  {Zhu}},\ }\bibfield  {title} {\enquote {\bibinfo {title} {A {Riemannian}
  conjugate gradient method for optimization on the {Stiefel} manifold},}\
  }\href {\doibase 10.1007/s10589-016-9883-4} {\bibfield  {journal} {\bibinfo
  {journal} {Computational optimization and Applications}\ }\textbf {\bibinfo
  {volume} {67}},\ \bibinfo {pages} {73} (\bibinfo {year} {2017})}\BibitemShut
  {NoStop}%
\end{thebibliography}%

\newpage
\appendix

\section{the QPV\topdf{$_\theta$}{-θ} protocol and a simple attack}
    \label{sec:qpvprot}

\subsection{Protocol description} \label{sec:protdesc}
What follows is a description of QPV$_{\theta}$, a family of protocols 
already introduced in~\cite{KMS11,LauLo11}.
In order to simplify the discussion, all parties are confined to 1D space, i.e.\ on a line; 
see~\cite{LauLo11,Unruh14} for extensions to D-dimensional space.
We further assume that the time needed to perform computations is negligible with respect
to the travel time of the signals.

Before the start of the protocol, a prover publicly claims to be at position $P$.
Two verifiers, who distrust the prover, would like to verify his claim.
They thus pick two trusted stations at position $V_1$ and $V_2$, respectively
to the left and to the right of $P$, and synchronize their clocks.
One round of the protocol then proceeds as follows (\cref{fig:diagQPVhonest}):
\begin{enumerate}
    \item 
    $V_1$ and $V_2$ agree on two random bits $x,b \in \{0,1\}$ by means of pre-shared
    randomness or through a secure classical channel.

    \item 
    $V_1$ prepares the state $\ket{\psi} = (R_{\theta})^{b}\ket{x}$, where 
    \begin{equation}
        R_\theta = \begin{pmatrix}
                       \cos\theta & -\sin\theta \\
                       \sin\theta & \cos\theta
                   \end{pmatrix}
    \end{equation}
    is a real rotation%
    \footnote{%
        We choose $\theta$ to be the polarization angle, at variance with the
        convention for a $\sigma_y$ rotation in the Bloch sphere where
        the corresponding angle would have been $\theta/2$.
    }
    defining a new encoding basis for $\ket{x}$.%
    \footnote{%
        While this might seems restrictive, for any pair of transformations
        $B_1$ and $B_2$ the verifier chooses to apply to the secret bit $\ket{x}$,
        we can always find an equivalent protocol with $B'_1 = I$ and $B'_2 = R_\theta$
        by setting $\cos(\theta) = \braket{0 | B_1^\dag B_2^{\phantom{\dag}} | 0}$,
        such that (w.l.o.g.)\ the four quantum inputs can be described by
        $b,x\in\{0,1\}$ as $\ket{\psi} = (R_\theta)^b \ket{x}$.
    }
    Then, $\ket{\psi}$ is sent towards $P$ through a public quantum channel.

    \item 
    $V_2$ sends $b$ towards $P$ through a public classical channel.
    The signals are carefully timed such that the quantum state and the classical bit
    arrive simultaneously at $P$, where the prover claims to be.
    For a moving prover, the signals are timed in his frame of reference.

    \item
    Upon receiving $\ket{\psi}$ and $b$, the prover applies $(R_\theta^{\dagger})^{b}$ to 
    $\ket{\psi}$ and measures in the computational basis, recovering $x$.
    He immediately broadcasts $x$ to $V_1$ and $V_2$.

    \item 
    The verifiers receive the results, checks their correctness
    and that the arrival timestamps of the signals
    are consistent with the minimal travel time allowed by relativity.
\end{enumerate}
The above steps are repeated for $N$ rounds.
The protocol terminates successfully if the answers to the challenges have been accepted
often enough, depending on a security parameter $\varepsilon>0$.
According to the precision of their clock, the verifiers can bound the prover's position
to a small neighborhood of $P$.

\subsection{Naïve security argument}
In order to fool the timestamp verification step, any attacker
that does \emph{not} control the neighborhood of $P$ would have to set up at least
two stations $A$ and $B$, respectively between $V_1$ and $P$ and between $V_2$ and $P$
(\cref{fig:diagQPVattack}).
When all inputs are classical, each attacker can copy its input and forward it
to the other;
then they can follow the honest prover's actions at both sites.
It becomes immediately clear that this strategy cannot work in the quantum case,
because the four possible states $\ket{\psi}$ sent by the verifier are not
in general all orthogonal to each other---except for a ``classical protocol'',
where $\theta = (0 \mod\fpi{2})$.
Thus the quantum input cannot be deterministically copied at $A$ and sent to $B$.
As it turns out, if $A$ and $B$ do not share entanglement (the \textit{No-PE model}),
the security of QPV$_\theta$ can indeed be rigorously proven~\cite{BCF+14}.

\subsection{Teleportation-based attack} \label{sec:telep_attack}
Nonetheless, for $\theta=\lpi4$ the protocol can be broken with unit
probability by a strategy involving pre-shared entangled states and the
teleportation protocol.
Already in his seminal paper~\cite{KMS11}, Kent shows that QPV$_\text{BB84}$
(analogous to QPV$_{\lpi{4}}$ within our notation) can be perfectly broken by exploiting
the commutation properties of the standard teleportation correction operators.
In this case, the basis $R_{\lpi{4}}$ chosen by the verifier has the peculiar property
that the honest prover's actions on the teleported state can be simulated at $B$
\emph{before} the end of the teleportation protocol, i.e.\ without waiting for the
usual classical information from $A$ telling $B$ which correction operator to apply
to his quantum state in order to get $\ket{\psi}$.
This bears some resemblance to how error correction works in measurement based
quantum computation (MBQC)~\cite{Raussendorf_2007}.
A way to see this is that the Pauli $X$, $Z$ and $XZ$ operators keep the sets
$\{\ket{0},\ket{1}\}$ and $\{H\ket{0}, H\ket{1}\}$ invariant.
The teleportation-based attack strategy is efficient and has been shown to be optimal
in~\cite{RG15}, consuming one ebit per round.

\section{QPV\topdf{$_\theta$}{-θ} in the INQC picture}
    \label{sec:INQC}

While our definition of the attack model is useful in the circuit picture that
we propose, an extensive part of the
literature~\cite{speelman2015instantaneous,BCF+14,BeigiKoenig11,GC20} on this topic
characterises the attacks through the (more general) INQC implementation of a
suitable unitary $U_\text{AB}$ on a bipartite quantum input $\rho_\text{AB}$
sent by the verifiers.
We establish here a mapping from our model to the INQC one.
In particular, we show how the results in~\cite{GC20} apply to ours.

During a round of QPV$_\theta$, the adversaries receive the
quantum-classical~\cite{wilde_2013} state:
\begin{equation}
    \ket{\Psi_b(x)}_\text{AB} =
        (R_\theta)^b \ket{x}_\text{A} \otimes \ket{b}_\text{B},
\end{equation}
where $x=0,1$ and $b=0,1$ with equal probability $1/4$.
If their goal was for just one of them to be able to obtain $\ket{x}$,
they could achieve it by implementing the controlled-unitary
\begin{equation}
    \tilde{U}_\theta =
        I \otimes \ket{0}\bra{0} + R_{-\theta} \otimes \ket{1}\bra{1},
\end{equation}
which leaves them with the state
\begin{equation}
    \ket{\tilde{\Psi}_b(x)}_\text{AB} = \ket{x}_\text{A} \otimes \ket{b}_\text{B}.
\end{equation}
Alice could then measure his system in the computational basis to obtain $x$.
However, at this point they have already used up the allowed round of
simultaneous communication to implement $\tilde{U}_\theta$, and Alice cannot
send $x$ to Bob.

The issue can be fixed by the addition of a CNOT gate controlled on
Alice's side:
\begin{equation} \label{eq:Utheta}
    U_\theta = \text{CNOT}_\text{AB} \, \tilde{U}_\theta
\end{equation}
leaving them with the state
\begin{equation}
    \ket{\Psi'_b(x)}_\text{AB} = \ket{x}_\text{A} \otimes \ket{x\oplus b}_\text{B}.
\end{equation}
Now Bob is also able to retrieve $x$, by measuring in the computational basis and
XORing the result with $b$ (which is also available as a classical bit).
An attack to QPV$_\theta$ can thus be equivalently described in terms of an
INQC implementation of the two-qubit unitary $U_\theta$.
An analogous argument applies to the QPV$_{(n)}$ protocol with $n$ bases defined in
\cref{sec:approxattacks}, giving a unitary $U_n$ acting on the $(2 \otimes n)$-%
dimensional state $\ket{\Psi_b(x)}_\text{AB} = (R_{\theta_b}\ket{x}) \otimes \ket{b}$
for $x\in\{0,1\}$, $b\in\integersbelow{n}$, $\theta_b \in S_n$.

The protocol defined in~\cite{GC20} gives an INQC implementation of all two-qubit
unitaries consuming a linear amount of ebits in the desired approximation accuracy.
Through the embedding defined above, all QPV$_\theta$ protocols can be attacked
in this way (but not QPV$_{(n)}$ when $n>2$).
Matching their notation, their strategy is based on the decomposition~\cite{Kraus2001}:
\begin{equation} \label{eq:kak_dec}
    U = (R_1 \otimes S_1) \Omega (R_2 \otimes S_2),
\end{equation}
where $R_i, S_i$ are single-qubit unitaries.
The matrix
\begin{equation}
    \Omega = \exp \{ i
    \left(
        \alpha\, \sigma_x \otimes \sigma_x +
        \beta\, \sigma_y \otimes \sigma_y +
        \gamma\, \sigma_z \otimes \sigma_z
    \right) \}
\end{equation}
describes the nonlocal part of $U$, and is diagonal in a basis of maximally
entangled states called \textit{magic basis}.
While in general their strategy fails with some error, they give in two special
cases a perfect implementation of $U$, provided that $\alpha, \beta, \gamma$
are all integer multiples of $\lpi{2^n}$:
\begin{itemize}
    \item
    If $n=2$, consuming 2 ebits;

    \item
    If $n>2$, consuming a finite number of ebits.
\end{itemize}
We can obtain the values of $\alpha_\theta, \beta_\theta, \gamma_\theta$ for our
$U_\theta$ through the Cartan (also known as $KAK$) decomposition~\cite{Drury_2008}.
Choosing:
\begin{equation}
    \begin{alignedat}{2}
        R_1 &= \frac{I - iZ}{\sqrt{2}}, &\qquad\quad
        S_1 &= R_{\lpi4}, \\
        R_2 &= R_{-\theta/2}, &
        S_2 &= \frac{Z - iI}{\sqrt{2}},
    \end{alignedat}
\end{equation}
we obtain a factorization of our $U_\theta$ in \cref{eq:Utheta} in the form of
\cref{eq:kak_dec}, with
\begin{align}
    \alpha_\theta = 0, &&
    \beta_\theta = \theta/2, &&
    \gamma_\theta = \pi/4.
\end{align}
For $\theta$ multiples of $\lpi{2^n}$ a direct application of their strategy
gives an exact attack consuming $4n + 15$ ebits.
The attacks we found in \cref{sec:exactattacks} require, for $n=2,3,4$,
respectively $1,2,3$ ebit(s);
moreover, we observe up to $d=12$ that an entangled qu$d$it of even dimension
is sufficient to attack $\theta$ multiples of $\lpi{2d}$, suggesting in this
case the existence of attacks for all $n$ requiring just $n-1$ ebits.
We emphasize though that the gain in ebit consumption is likely due to the
large amount of structure in the family $U_\theta$ that we consider, motivated
by the analysis of simple protocols with only one quantum input, while the attacks
in~\cite{GC20} work for all two-qubit unitaries.

\section{graphical no-go proofs for \topdf{$d=2$ and $d=3$}{d=2 and d=3} exact attacks}
    \label{sec:graphproofs}

We introduce here a representation of the $2d$-dimen\-sio\-nal Hilbert space
to which the states
\begin{equation} \tag{\ref{eq:psi_bxs}}
    \ket{\psi_b(x,s)} = V(R_{\theta}\otimes U)^b (\ket{x}\otimes\ket{s})
\end{equation}
belong for all $b,x,s$.
The visualization is loosely based on \textit{hypergraphs}, a generalization of graphs
where an edge is allowed to join any number of vertices.
Additionally, our edges are labelled---in particular, two different edges can
join the same subset of vertices.
\Cref{fig:hypergraphs} in the main text, as well as \crefrange{gpair:d21}{gpair:d32}
in the proof below, are examples of the visualization:
\begin{itemize}
    \item 
    Each vertex represents an element $\ket{u}$ of the computational basis.
    An edge joining a specific subset of vertices represents a state
    $\ket{\psi_b(x,s)}$ having that subset as its support.
    
    \item 
    Each hypergraph has $2d$ vertices and $d$ edges in the inner (outer)
    region, where $x=0$ ($x=1$).
    
    \item 
    The color of the edges encodes the different values of the index $s$.
    
    \item 
    To an attack strategy corresponds a pair of graphs, one for $b=0$ and one for $b=1$.

    \item
    Some information about the state is lost, e.g.\ the actual amplitudes
    $\braket{u | \psi_b(x,s)}$ are not represented.
\end{itemize}

We can easily characterize the inner products between the states in \cref{eq:psi_bxs}:
for all $b,x,y\in\{0,1\}$ and $s,t\in\integersbelow{d}$,
\begin{align}
    \label{eq:inner_delta}
    \braket{\psi_b(x,s)|\psi_b(y,t)}&=\delta_{xy}\delta_{st},\\
    \label{eq:inner_RU}
    \braket{\psi_0(x,s)|\psi_1(y,t)}&=\braket{x|R_{\theta}|y}\braket{s|U|t}.
\end{align}
In particular, the unitarity of $R_\theta,U,V$ implies that for fixed $b$
the vectors $\ket{\psi_b(x,s)}$ span the whole output space.
For an attack strategy to be exact it has to satisfy the DDC:
\begin{align} \tag{\ref{eq:DDC}}
    \braket{u|\psi_b(0,s)}=0 && \text{or} && \braket{u|\psi_b(1,s)}=0,
\end{align}
for all $u\in\integersbelow{2d}$, $s\in\integersbelow{d}$, $b\in\{0,1\}$.
Phrased geometrically, the states $\ket{\psi_b(0,s)}$ and $\ket{\psi_b(1,s)}$
need to have disjoint supports in the computational basis.
\Cref{eq:DDC} has another interesting consequence: either the
protocol is classical and trivially broken (i.e.\ $\theta\equiv 0\mod\fpi2$),
or the states $\ket{\psi_b(x,s)}$ have to satisfy
\begin{align} \label{eq:balanced}
    \sum_s|\braket{u|\psi_b(x,s)}|^2=\frac12, &&
    \begin{aligned}
        &\forall\, b,x\in\{0,1\},\\
        &\forall u\in\integersbelow{2d}.
    \end{aligned}
\end{align}
This last equation, proven in \cref{sec:proofbalanced}, has further implications
on the support of the states $\ket{\psi_b(x,s)}$.

The restrictions imposed by \crefrange{eq:inner_delta}{eq:balanced} can be captured
by the following necessary (but not sufficient) structure on the graphs:
\begin{enumerate}[(I)]
    \item \label[prop]{prop:disjsupport}
    Disjointness of $\ket{\psi_b(0,s)}$'s and $\ket{\psi_b(1,s)}$'s supports
    implies that any vertex that is part of an $s$-colored inner edge cannot
    \emph{also} be part of the corresponding $s$-colored outer edge.

    \item \label[prop]{prop:sum12}
    \Cref{eq:balanced}, giving the total ``probability budget'' for the inner (outer)
    edges crossing a given vertex, has several graphical implications:
    \begin{enumerate}
        \item 
        All vertices have to be part of at least one inner and one outer edge

        \item
        Vertices joined by an inner (outer) edge of length 2 cannot be part of
        other inner (outer) edges.

        \item \label[prop]{prop:atleast2}
        Each edge has to join at least two vertices
        and, due to \cref{prop:disjsupport}, cannot join more than $2d-2$ vertices.
    \end{enumerate}

    \item \label[prop]{prop:nodedges}
    Due to \cref{prop:disjsupport,prop:sum12}, no vertex can be covered by all
    inner edges or by all outer edges.

    \item \label[prop]{prop:nojustone}
    Per \cref{eq:psi_bxs},
    for fixed $b$, all states $\ket{\psi_b(x,s)}$ are orthogonal to each other.
    This forbids any two edges from having only one vertex in common.

    \item \label[prop]{prop:crossprop}
    Finally, while a bit trickier to visualize, \cref{eq:inner_RU} imposes that
    if an $s$-colored edge on the left graph does not share any vertex with one
    of the $t$-colored edges on the right graph, then all four edges of that color
    combination $(s,t)$ represent orthogonal states.
\end{enumerate}

Using \crefrange{prop:disjsupport}{prop:crossprop} we proceed to prove the following
result, already present in~\cite{LauLo11}.

\begin{theorem}
    Under the assumptions of our attack model (\cref{sec:protocol}), adversaries sharing
    a maximally entangled qubit cannot perfectly break QPV$_\theta$ unless $\theta$ is a
    multiple of $\lpi{4}$.
    A maximally entangled qutrit gives them even less power:
    they can perfectly break only classical protocols, namely $\theta$ multiple of $\lpi2$.
\end{theorem}

\begin{proof}
\allowdisplaybreaks
For $d=2$, we can quickly rule out most graphs.
As a matter of fact, for either one of the graphs forming the pair describing
the attack, \cref{prop:atleast2} leaves two possibilities:
\begin{align*}
    \mhgraph[rotate=-45]{2}{{1,2}/1,{3,4}/1}{{3,4}/1,{1,2}/1} &&
    \mhgraph[rotate=-45]{2}{{1,2}/1,{3,4}/1}{{1,4}/1,{2,3}/1} 
\end{align*}
but \cref{prop:disjsupport} rules out the second one.
Up to vertex reordering, two graph pairs are possible:
\begin{align} \label[gpair]{gpair:d21}
   &\begin{tabular}{cc|cc}
       $b=0$ &\hspace{5pt} &\hspace{5pt} & $b=1$ \\
       \mhgraph[rotate=-45]{2}{{1,2}/1,{3,4}/1}{{3,4}/1,{1,2}/1} &&&
       \mhgraph[rotate=-45]{2}{{1,2}/1,{3,4}/1}{{3,4}/1,{1,2}/1} 
    \end{tabular}
 \\[6pt]      \label[gpair]{gpair:d22}
   &\begin{tabular}{cc|cc}
       $b=0$ &\hspace{5pt} &\hspace{5pt} & $b=1$ \\
       \mhgraph[rotate=-45]{2}{{1,2}/1,{3,4}/1}{{3,4}/1,{1,2}/1} &&&
       \mhgraph[rotate=45]{2}{{1,2}/1,{3,4}/1}{{3,4}/1,{1,2}/1}
    \end{tabular}
\end{align}
but \cref{prop:crossprop} on \cref{gpair:d21} implies that the four
top edges (inner blue and outer orange on both graphs) would correspond to four
orthogonal states defined on the same support of size $2$.
Thus only \cref{gpair:d22} is a viable option.
\Cref{eq:balanced} then shows that in this configuration
\begin{equation*}
    |\braket{\psi(0,x,s)|\psi(1,y,t)}| = 
    |\braket{x|R_\theta|y} \braket{s|U|t}| = \frac12
\end{equation*}
for all $x,y,s,t$.
This implies $\theta=\fpi[n]{4}$ and $U = R_{\lpi{4}}$ up to phases;
we thus recover the result in~\cite{LauLo11} that adversaries sharing an
entangled qubit can only break QPV$_{\lpi{4}}$.

For $d=3$ we start focusing only on the inner edges at fixed $b$.
We have to place three edges of length ranging from 2 to 4, while satisfying
\crefrange{prop:disjsupport}{prop:nojustone}.
We are left with six possibilities:
\begin{align*}
    \mhgraph{3}{{1,2}/1,{3,4}/1,{5,6}/1}{} &&
    \mhgraph{3}{{1,2}/1,{3,4,5}/1,{4,5,6}/2}{} &&
    \mhgraph{3}{{1,2}/1,{3,4,5}/1,{3,4,5,6}/2}{} \\[5pt]
    \mhgraph{3}{{1,2}/1,{3,4,5,6}/1,{3,4,5,6}/2}{} &&
    \mhgraph{3}{{1,2,3}/1,{2,3,4,5}/2,{4,5,6}/1}{} &&
    \mhgraph{3}{{1,2,3,4}/1,{3,4,5,6}/2,{5,6,1,2}/3}{}
\end{align*}
When placing the outer edges (subject to the same rules), we have to be careful
to respect \cref{prop:sum12,prop:nojustone}.
We are left with two non-trivial graphs:
\begin{align*}
    \mhgraph{3}{{1,2}/1,{3,4}/1,{5,6}/1}{{3,4}/1,{5,6}/1,{1,2}/1} &&
    \mhgraph{3}{{1,2,3,4}/1,{3,4,5,6}/2,{5,6,1,2}/3}{{5,6}/1,{1,2}/1,{3,4}/1}
\end{align*}
We can rule out the second case by noticing that it implies that three states
(e.g.\ inner green and orange, outer blue) are all orthogonal on a space of
dimension 2, namely the intersection of their supports.
Drawing the possible full pairs based on the first case, we obtain:
\begin{align} \label[gpair]{gpair:d31}
   &\begin{tabular}{c|c}
        $b=0$ & $b=1$ \\
        \mhgraph{3}{{1,2}/1,{3,4}/1,{5,6}/1}{{3,4}/1,{5,6}/1,{1,2}/1} &
        \mhgraph{3}{{1,2}/1,{3,4}/1,{5,6}/1}{{3,4}/1,{5,6}/1,{1,2}/1}
    \end{tabular}
 \\[6pt]      \label[gpair]{gpair:d32}
   &\begin{tabular}{c|c}
        $b=0$ & $b=1$ \\
        \mhgraph{3}{{1,2}/1,{3,4}/1,{5,6}/1}{{3,4}/1,{5,6}/1,{1,2}/1} &
        \mhgraph{3}{{2,3}/1,{4,5}/1,{6,1}/1}{{4,5}/1,{6,1}/1,{2,3}/1}
    \end{tabular}
\end{align}
but \cref{prop:crossprop} applied to \cref{gpair:d31,gpair:d32} tells us that
either we have again too many orthogonal states on a support of size two, or that two
states are orthogonal over a support intersection of size one.
In both cases we have a contradiction, so there is \emph{no exact attack} to 
QPV$_\theta$ for any nontrivial $\theta$ for adversaries sharing an entangled qutrit,
confirming the result in~\cite{LauLo11} and proving the theorem.
\end{proof}

Unfortunately, this method shows its limitations when applied to the $d=4$ case.
With the help of a program, we enumerated all graphs where
\crefrange{prop:disjsupport}{prop:nojustone} are satisfied, along with some more
refined conditions.
For fixed $b$, we were able to get the possible configurations of inner edges down to
about a thousand, and we could single out 18 of them that admit at least one compliant
set of outer edges.
However the number of valid pairs that can be produced with them is too big to handle
manually, even when using \cref{prop:crossprop}.
A more careful analysis of the DDC could give tighter rules, allowing to reduce the
possible pairs to a manageable number.

\section{proof of \topdf{\cref{eq:balanced}}{eq.~(\ref{eq:balanced})}}
    \label{sec:proofbalanced}

We now prove that when \cref{eq:DDC} (the DDC) is imposed on a non-classical
protocol, namely $\theta\not\equiv0\pmod{\fpi2}$, the states $\ket{\psi_b(x,s)}$
defined in \cref{eq:psi_bxs} satisfy the relation \eqref{eq:balanced}, recalled below
\begin{equation} \tag{\ref{eq:balanced}}
    \sum_s | \braket{u|\psi_b(x,s)}|^2 = \frac12
\end{equation}
for all $b,x\in\{0,1\}$ and $u\in\integersbelow{2d}$.

\begin{proof}
The families of vectors $\{\ket u\}$, $\{\ket{\psi_0(x,s)}\}$ and $\{\ket{\psi_1(y,t)}\}$
form three orthonormal bases of the same space of dimension $2d$.
We can expand $\ket{\psi_1(y,t)}$ in the $\{\ket{\psi_0(x,s)}\}$ basis, obtaining
$\forall y\in \{0,1\}, \forall t \in \integersbelow{d}$:
\begin{align} \label{eq:expandpsi}
    \ket{\psi_1(y,t)} \notag
        &= \sum_{x,s}{\braket{\psi_0(x,s)|\psi_1(y,t)}\ket{\psi_0(x,s)}} \\\notag
        &= \sum_{x,s}\left(\bra{x}\otimes\bra{s}\right) V^\dagger
                V (R_\theta\ket{y} \otimes U\ket{t}) \ket{\psi_0(x,s)} \\
        &= \sum_{x,s}{\braket{x|R_{\theta}|y}\braket{s|U|t}\ket{\psi_0(x,s)}},
\end{align}
where we used \cref{eq:inner_RU} in the second step.
For brevity, we define the scalar $\psi_{u,b}(x,s)\eqdef\braket{u|\psi_b(x,s)}$.
The DDC can thus be seen as imposing
\begin{align} \label{eq:proofDDC}
    \psi_{u,b}(0,s) \, \psi_{u,b}^*(1,s) = 0,
\end{align}
$\forall u\in\integersbelow{2d}, \forall s \in\integersbelow{d}$ and $\forall b\in\{0,1\}$.
Projecting \cref{eq:expandpsi} onto $\ket{u}$, for all $y\in\{0,1\}$:
\begin{align}
    & \ \psi_{u,1}(y,t) \notag \\
        & = \sum_{x,s} \braket{x|R_\theta|y} \braket{s|U|t} \psi_{u,0}(x,s) \\\notag
        & =  \sum_s \braket{s|U|t} \left[ \braket{0|R_\theta|y} \psi_{u,0}(0,s)
            + \braket{1|R_\theta|y} \psi_{u,0}(1,s) \right] ,
\end{align}
substituting $y=0$ and $y=1$:
\begin{align}
    & \ \psi_{u,1}(0,t) \label{eq:suby0} \\\notag
       & = \sum_s \braket{s|U|t} \left[ \cos(\theta) \psi_{u,0}(0,s)
            + \sin(\theta) \psi_{u,0}(1,s) \right] , \\[2pt]
    & \ \psi_{u,1}(1,t) \label{eq:suby1} \\\notag
       & = \sum_s \braket{s|U|t} \left[ \cos(\theta) \psi_{u,0}(1,s)
            - \sin(\theta) \psi_{u,0}(0,s) \right] .
\end{align}
Using the DDC~(\ref{eq:proofDDC}) for $b=1$
\begin{equation}
    \psi_{u,1}^{\phantom{x}}(0,s) \, \psi_{u,1}^*(1,s) = 0,
\end{equation}
along with \cref{eq:suby0,eq:suby1}, we obtain:
\begin{equation}
\begin{aligned}
    0 = \sum_{s,s'} \ & \braket{s|U|t} \braket{s'|U|t}^* \\[-10pt]
        \cdot & \left[ \cos(\theta) \psi_{u,0}^{\phantom{x}}(0,s)
            + \sin(\theta) \psi_{u,0}^{\phantom{x}}(1,s) \right] \\[2pt]
        \cdot & \left[ \cos(\theta) \psi_{u,0}^*(1,s')
            - \sin(\theta) \psi_{u,0}^*(0,s') \right].
\end{aligned}
\end{equation}
Summing over $t$ gives:
\begin{equation}
\begin{aligned}
    0 = \sum_{s,s'} \ &
            \Big( \sum_t \braket{s|U|t} \braket{t|U^\dagger|s} \Big) \\[-5pt]
        \cdot & \left[ \cos(\theta) \psi_{u,0}^{\phantom{x}}(0,s)
            + \sin(\theta) \psi_{u,0}^{\phantom{x}}(1,s) \right] \\[2pt]
        \cdot & \left[ \cos(\theta) \psi_{u,0}^*(1,s')
            - \sin(\theta) \psi_{u,0}^*(0,s') \right],
\end{aligned}
\end{equation}
and since $\sum_t{\ket t\!\bra t} = I$ and $\braket{s|s'} = \delta_{ss'}$,
we have
\begin{equation}
\begin{aligned}
    0 = \sum_s \ & \left[ \cos(\theta) \psi_{u,0}^{\phantom{x}}(0,s)
                + \sin(\theta) \psi_{u,0}^{\phantom{x}}(1,s) \right] \\[-6pt]
            \cdot & \left[ \cos(\theta) \psi_{u,0}^*(1,s)
                - \sin(\theta) \psi_{u,0}^*(0,s) \right].
\end{aligned}
\end{equation}
With the DDC for $b=0$, this simplifies to
\begin{equation*}
    \cos(\theta)\sin(\theta)
    \Big( \sum_s |\psi_{u,0}(1,s)|^2 - \sum_s |\psi_{u,0}(0,s)|^2 \Big) = 0 \,.
\end{equation*}

Now we make use of the assumption of non-classicality of the protocol, namely
$\theta\not\equiv 0\pmod{\fpi2}$, ensuring thus $\cos(\theta)\sin(\theta)\neq0$.
With the previous equation and the normalization
of $\ket{u}$, expressed in the $\{\ket{\psi_0(x,s)}\}$ basis:
\begin{equation}
    \sum_s |\psi_{u,0}(0,s)|^2 + \sum_s |\psi_{u,0}(1,s)|^2 = 1,
\end{equation}
it implies that
\begin{equation}
    \sum_s |\psi_{u,0}(0,s)|^2 = \sum_s |\psi_{u,0}(1,s)|^2 = \frac12 .
\end{equation}
Writing now $\ket{\psi_0(x,s)}$ in the $\{\ket{\psi_1(y,t)}\}$
basis and following the whole calculation again, we finally obtain
\begin{equation}
    \sum_s |\psi_{u,b}(x,s)|^2=\frac12
\end{equation}
for all $u\in\integersbelow{2d}$ and $b,x\in\{0,1\}$,
when $\theta\not\equiv 0\pmod{\fpi2}$.
\end{proof}

\section{numerical methods}
    \label{sec:num_methods}

\subsection{Exact attacks} \label{sec:num_methods_exact}
The DDC can be equivalently written as:
\begin{equation} \label{eq:DDC_prod}
    \braket{u|\psi_b(0,s)}\braket{\psi_b(1,s)|u}=0.
\end{equation}
We can obtain polynomial equations from \cref{eq:DDC_prod} using the
definition of the states $\ket{\psi_b(x,s)}$ in \cref{eq:psi_bxs},
\begin{equation} \tag{\ref{eq:psi_bxs}}
    \ket{\psi_b(x,s)}\eqdef V(R_{\theta}\otimes U)^b (\ket{x}\otimes\ket{s}),
\end{equation}
by writing $U_{st}$ for $\braket{s|U|t}$, $V_{u,xs}$ for
$\bra{u} V (\ket{x}\otimes\ket{s})$ and $R_{xy}$ for $\braket{x|R_\theta|y}$:
\begin{align}
    \label{eq:polyDDC0}
    0 &= V_{u,0s}^* V_{u,1s}^{\phantom{*}}\,, \\
    \label{eq:polyeq:DDC}
    0 &= \Big(\sum_{ij} V_{u,ij} R_{i0} U_{js}\Big)
         \Big(\sum_{kl} V_{u,kl}^* R_{k1}^{\phantom{*}} U_{ls}^*\Big).
\end{align}
Matrices $U$ and $V$ solving the above equations describe an attack if
and only if they are unitary; they have therefore to also satisfy:
\begin{align}
    \label{eq:Uunitarity}
    &\sum_{k=1}^{d} U_{k,i}^* U_{k,j}^{\phantom{*}} = \delta_{ij} &&
            \forall i,j \in \integersbelow{d}, i\geq j , \\
    \label{eq:Vunitarity}
    &\sum_{k=1}^{2d} V_{k,i}^* V_{k,j}^{\phantom{*}} = \delta_{ij} &&
            \forall i,j \in \integersbelow{2d}, i\geq j .  
\end{align}
Strictly speaking, these constraints are not polynomial equations as they
involve complex conjugation;
still, they can be expressed as polynomials in the real and imaginary parts of
$U$ and $V$'s entries.
In order to simplify the problem and reduce the number of variables, we restrict
our search to attacks that can be written as real orthogonal matrices.
Results on the approximate attacks in \cref{sec:approxattacks} suggest that
this is not too restrictive, as detailed in \cref{foot:ortho}.
Assuming real variables, \cref{eq:polyDDC0,eq:polyeq:DDC,eq:Uunitarity,eq:Vunitarity}
consist of
\begin{equation}
    2d^2 +2d^2 + \frac{d(d+1)}{2} + \frac{2d(2d+1)}{2} = \frac{13d^2 + 3d}{2}
\end{equation}
equations of the form $f_i(U,V)=0$ in $d^2 + (2d)^2 = 5d^2$ scalar variables,
parametrized by $\theta$.
It should be noted however that not all these constraints are independent:
for example, \cref{eq:polyDDC0} already implies the orthogonality of
$d$ columns of $V$.

While computational algebraic tools for working with symbolic polynomial
equations are available, their inherent exponential scaling
makes them challenging to apply directly to our system, for which the smallest
interesting case (real matrices, $d=4$) involves 80 variables and 174 equations.
For example, techniques involving \textit{sum of squares} (SOS)
proofs~\cite{parrilo2003semidefinite} construct a hierarchy of SDPs such that
any feasible point provides a certificate of unsolvability of
the system (and is guaranteed to be found at some level of the hierarchy).
It should be noted that such approaches can be unsuccessful when applied \textit{as is},
and may prove more effective by taking some of our problem's symmetries into account,
e.g.\ the freedom in assigning an order to the states labeled by $x,s$ and
to the basis elements $\{\ket{u}\}$.

Even if the system looks heavily overdetermined, we know it has at least a trivial
solution for all $d$, namely when the protocol is classical ($\theta = 0\mod\fpi2$).
In the following, we will tacitly ignore those.
In order to look for other solutions, we define:
\begin{equation}\label{eq:sqsum}
    F = \sum_{i} f_i^2,
\end{equation}
i.e.\ the sum of the squares of all polynomials.
The zeros of the function in \cref{eq:sqsum} are also simultaneous zeros of all
the polynomials $f_i$.
We can minimize $F$ with a numerical method;
if we find zero as minimum, we have found an exact attack.
In order to look for zeros of $F$, we leveraged a nonlinear least-squares method
provided by the Python library SciPy~\cite{2020SciPy-NMeth}.
We could find many new solutions, up to $d=12$:
the results are collected in \cref{tab:attacksfound} of \cref{sec:exactattacks}.

\subsection{Explicit solutions} \label{sec:explicit_sols}

The following block matrices $U$ and $V$ are examples of explicit exact attacks
we found via inspection of the results of the numerical optimization.

For $d=4$, $\theta = \lpi8$, we provide two non-trivially equivalent attacks:
\begin{align}
    V &= \frac12
    \begin{pmatrix}
         X  &  I  & -Z  & ZX  \\
        ZX  &  X  &  I  &  Z  \\
         X  & -I  & -Z  &-ZX  \\
        ZX  & -X  &  I  & -Z
    \end{pmatrix}
    \\
    U &= \frac{1}{\sqrt2}
    \begin{pmatrix}
        R_{-\lpi8}  &  R_{\lpi8}\, Z \\
        Z R_{\lpi8} &  R_{-\lpi8}
    \end{pmatrix},
\end{align}
and
\begin{align}
    V &= \frac{1}{\sqrt2}
    \begin{pmatrix}
        XHX &  0  &  0  &  I  \\
         0  &-XHX & -I  &  0  \\
        ZX  &  0  &  0  &  H  \\
         0  & -ZX & -H  &  0
    \end{pmatrix}
    \\
    U &= \frac{1}{\sqrt2}
    \begin{pmatrix}
        R_{\lpi8}  &  R_{-\lpi8}\, Z \\
        Z R_{-\lpi8} &  R_{\lpi8}
    \end{pmatrix}.
\end{align}
For $d=6$, $\theta = \lpi{12}$, we found:
\begin{align}
    V &= \frac12
    \begin{pmatrix}
          I \otimes \sqrt{2} R_{\lpi{6}}  & \begin{matrix} 0 \\ 0 \end{matrix}
        & X \otimes \sqrt{2} R_{-\lpi{3}} & \begin{matrix} 0 \\ 0 \end{matrix} \\
           H \otimes I  & \begin{matrix} -Z \\ X \end{matrix}
        & -ZH \otimes X & \begin{matrix} -X \\ Z \end{matrix} \\
           H \otimes I  & \begin{matrix} Z \\ -X \end{matrix}
        & -ZH \otimes X & \begin{matrix} X \\ -Z \end{matrix}
    \end{pmatrix}
    \\
    U &=
    \begin{pmatrix}
        A & -\frac12 ZH & \frac1{3-\sqrt3} ZX \\
        B &  \frac12 ZH & \frac1{3+\sqrt3} ZX \\
        \frac1{\sqrt6} R_{\lpi{12}} & \frac1{\sqrt2} R_{\lpi{12}} & \frac1{\sqrt3} ZH
    \end{pmatrix},
\end{align}
where
\begin{align}
    A &= \frac{2-\sqrt{3}}{\sqrt{24}} ZX - \frac1{2\sqrt{2}} I, \\
    B &= \frac1{2\sqrt{2}} XZ - \frac{2+\sqrt{3}}{\sqrt{24}} I.
\end{align}
The above is a special case of a continuum of solutions with one real degree of freedom.

\subsection{Approximate attacks} \label{sec:num_methods_approx}
Using the notation of \cref{sec:circuitrep}, we have that $\forall x,b$ a measurement
result of $s,u$ occurs with probability
\begin{equation}
    p(x,b,s,u) = | \braket{u | \psi_b(x,s)} |^2 \,p(x)\, p(b)\, p(s).
\end{equation}
Alice and Bob's best guess for the value of $x$ is thus
\begin{equation}
    p_\text{succ}(b,s,u) = \max_x{\,p(x,b,s,u)},
\end{equation}
their probability of error being $\perr = 1-p_\text{succ}$.
As our protocol only involves qubits, namely $x \in \{0,1\}$, we have
\begin{equation}
    \perr(b,s,u) = \min\{p(0,b,s,u), p(1,b,s,u)\}.
\end{equation}
The overall error probability for an attack strategy can be obtained by summing over
$b,s,u$ and remembering that $x,b$ and $s$ are uniformly distributed:
\begin{equation} \label{eq:p_err2}
    \perr = \frac{1}{2\cdot 2\cdot d} \sum_{b,s,u}
        \min\big\{ |\braket{u|\psi_b(0,s)}|^2 , |\braket{u|\psi_b(1,s)}|^2 \big\}.
\end{equation}
Imposing the DDC (eq.~\ref{eq:DDC}) we have $\perr = 0$, as expected of exact
attacks.

From the point of view of the numerical optimization, $\perr$ is a function
of $\theta$ and of the (unitary%
\footnote{\label{foot:ortho}%
    At variance with the exact attacks in \cref{sec:exactattacks}, we have to
    explore the whole (complex) unitary space for us to obtain sensible bounds
    over a range of parameters, as only looking to orthogonal matrices
    makes little sense from a security standpoint.
    Nonetheless, when restricting to the (much faster to optimize) orthogonal group
    we obtain the exact same results and curves;
    this may be due to symmetries in our attack model---for example, $(U,V)$ and
    $(U^*,V^*)$ are both attacks with the same $\perr$.
})
matrices $(U,V)$ defining an attack.
We seek to minimize it over all attack strategies at fixed $d$:
\begin{equation}
    \perr(\theta) = \min_{U,V}\, \perr(U,V,\theta).
\end{equation}
    Symmetries allow to restrict the relevant values of $\theta$ to $[0,\fpi{4}]$,
    through the
    relations $R_{\fpi2-\theta} = X R_\theta Z$ and $R_{-\theta} = X R_\theta X$,
    along with similar ones for the other quadrants.
    Extra $X$ and $Z$ are either absorbed into $V$ or taken into account
    by the adversaries by flipping the bit $x$.
At variance with the exact attacks, we want to obtain evidence about the
\emph{global} minimum of the continuous function~\eqref{eq:p_err2}.
Scalable (non-convex) numerical methods can only provide local minima~\cite{PARDALOS198833},
but we can repeat the optimization with thousands of uniformely sampled starting points,
keeping the best optimum (a strategy sometimes referred to as \textit{multistart}).
The shape of the search space can heavily affect the effectiveness of this method,
so we carry out the optimization in three different ways:
\begin{itemize}
    \item
    A constrained sparse interior point method~\cite{wachter2006implementation} (IPOPT),
    by imposing the unitarity constraints in \cref{eq:Uunitarity,eq:Vunitarity}.

    \item
    An unconstrained quasi-Newton method~\cite{BLNZ95} (L-BFGS), by parametrizing
    $U$ and $V$ with the skew-hermitian matrices $A_U$ and $A_V$ through either:
    \begin{itemize}
        \item
        the Cayley transform~\cite{Cayley1846,zhu2017riemannian}
            \begin{align*}
                \begin{split}
                    U &= (I + A_U)^{-1}(I - A_U), \\
                    V &= (I + A_V)^{-1}(I - A_V),
                \end{split}
            \end{align*}
        (with some care about the exceptional points)

        \item
        the exponential map $U = e^{A_U}, V = e^{A_V}$.
    \end{itemize}
\end{itemize}
When possible, the analytical gradient of $\perr$ with respect to the optimization
variables is obtained in order to speed up the computation.
The three methods give comparable results, requiring for example between $10^4$ and $10^5$
starting points for $d=4$ in order to converge to the same optimum.
The results of the minimization are presented in \cref{sec:approxattacks}.

\end{document}